\documentclass{sig-alternate-10}
\usepackage{amsmath,amsfonts,amssymb}
\usepackage{stmaryrd}
\pdfoutput=1

\newcommand{\bI}{{\bf I}}
\newcommand{\bV}{{\bf V}}
\newcommand{\dom}[1]{\text{dom}(#1)}
\newcommand{\vars}[1]{\text{vars}(#1)}

\newcommand{\AND}{\text{AND}}
\newcommand{\UNION}{\text{UNION}}
\newcommand{\OPT}{\text{OPT}}

\newcommand{\PSPACE}{\text{PSPACE}}
\newcommand{\ptime}{\text{PTIME}}
\newcommand{\coNP}{\text{coNP}}
\newcommand{\Wone}{\text{W[1]}}

\newcommand{\wdeval}{\text{{\sc wdEVAL}}}
\newcommand{\emb}{\text{{\sc EMB}}}
\newcommand{\pat}[1]{\text{pat}(#1)}
\newcommand{\ctw}[1]{\text{ctw}(#1)}
\newcommand{\tw}[1]{\text{tw}(#1)}
\newcommand{\dw}[1]{\text{dw}(#1)}
\newcommand{\bw}[1]{\text{bw}(#1)}

\newcommand{\tsupp}{\text{sp}}
\newcommand{\tbr}{\text{br}}
\newcommand{\supp}[1]{\text{supp}(#1)}

\newcommand{\C}{{\cal C}} 
\newcommand{\V}{{\cal V}} 
\newcommand{\B}{{\cal B}} 
\newcommand{\G}{{\cal G}} 
\newcommand{\HH}{{\cal H}} 
\newcommand{\ES}{{\cal S}} 
 
\newcommand{\T}{{\cal T}} 
\newcommand{\F}{{\cal F}} 
\newcommand{\PP}{{\cal P}} 
\newcommand{\CA}[1]{{\bf CA}(#1)} 
\newcommand{\VCA}[1]{{\bf VCA}(#1)} 
\newcommand{\GtG}[1]{{\bf GtG}(#1)} 
\newcommand{\wdpf}{{\bf wdpf}} 

\newcommand{\sem}[2]{\llbracket #1 \rrbracket_{#2}}

\newtheorem{example}{Example}
\newtheorem{theorem}{Theorem} 
\newtheorem{proposition}{Proposition}
\newtheorem{corollary}{Corollary}  
\newtheorem{definition}{Definition}

\newtheorem{lemma}{Lemma}

\newcommand{\OMIT}[1]{}

\numberofauthors{1}
\author{
Miguel Romero\\
\affaddr{University of Oxford, UK}\\
\email{miguel.romero@cs.ox.ac.uk}
}

\begin{document} 

\title{The tractability frontier of well-designed SPARQL queries\thanks{This project has received funding from the European
Research Council (ERC) under the European Union's Horizon 2020 research and
innovation programme (grant agreement No 714532). The paper reflects only the
authors' views and not the views of the ERC or the European Commission. The
European Union is not liable for any use that may be made of the information
contained therein.}}
\maketitle

\begin{abstract}
We study the complexity of query evaluation of SPARQL queries. We 
focus on the fundamental fragment of well-designed SPARQL restricted to the AND, OPTIONAL and UNION operators. 
Our main result is a structural characterisation of the classes of well-designed queries that can be evaluated in polynomial time. 
In particular, we introduce a new notion of width called \emph{domination width}, which relies on 
the well-known notion of treewidth. We show that, under some complexity theoretic assumptions, 
the classes of well-designed queries that can be evaluated in polynomial time are precisely those of bounded domination width.   

\end{abstract}

\section{Introduction}
\label{sec:intro}

The \emph{Resource Description Framework} (RDF) \cite{RDF} is the W3C standard for representing linked data on the Web. 
In this model, data is represented as \emph{RDF graphs}, which consist of collections of triples of \emph{internationalised resource identifiers} (IRIs). 
Intuitively, such a triple $(s,p,o)$ represents the fact that a \emph{subject} $s$ is connected to an \emph{object} $o$ via a \emph{predicate} $p$.

SPARQL \cite{SPARQL} is the standard query language for RDF graphs. 
In a seminal paper, P\'erez et al. \cite{PAG09} (see also \cite{PAG06}) gave a clean formalisation of the language, which laid the foundations for its theoretical study. 
Since then, a lot of work has been done in different aspects of the language such as query evaluation \cite{LM13,ACP12,KRU15,BPS15,KPS16}, optimisation \cite{lete,PS14,KRRV15}, 
and expressive power \cite{AG08,PW13,KRU15,ZB15,KK16}, to name a few. 

As shown in \cite{PAG09}, it is PSPACE-complete to evaluate SPARQL queries. 
This motivated the introduction of a natural fragment of SPARQL called the \emph{well-designed} fragment, whose evaluation problem 
is coNP-complete \cite{PAG09}. 
More formally, the \emph{evaluation problem} $\wdeval$ for well-designed SPARQL is to decide, given a well-designed query $P$, 
and RDF graph $G$ and a mapping $\mu$, whether $\mu$ belongs to the answer $\sem{P}{G}$ of $P$ over $G$. 
By now the well-designed fragment is central in the study of SPARQL 
and a lot of efforts has been done by the theory community to understand fundamental aspects of this fragment (see e.g. \cite{PAG09,lete,PS14,BPS15,KRRV15,KK16,KPS16,KRU15}). 
In this paper, we focus on the core fragment of well-designed SPARQL restricted to the AND, OPTIONAL and UNION operators, as defined in \cite{PAG09}.

Despite its importance, several basic questions remain open for well-designed SPARQL. 
As first observed in \cite{lete}, while the problem $\wdeval$ is coNP-complete, it becomes \emph{tractable}, i.e. polynomial-time solvable, 
for restricted classes of well-designed queries. Indeed, it was shown that $\wdeval$ is in PTIME for every class $\C$ of queries satisfying a certain \emph{local tractability} condition \cite{PAG09}. 
We emphasise that the above-mentioned result is briefly discussed in \cite{PAG09} as the focus of the authors is on the static analysis and optimisation of queries rather than complexity of evaluation. 
Subsequent works \cite{BPS15,KPS16} have studied the complexity of evaluation in more depth but the focus has been mainly on the fragment of SPARQL including the SELECT operator (i.e., \emph{projection}). 
In particular, the following fundamental question regarding the core well-designed fragment remains  open:
\emph{which classes $\C$ of well-designed SPARQL can be evaluated in polynomial time?}

Our main contribution is a complete answer to the question posed above. In particular, 
we  introduce a new width measure for well-designed queries called \emph{domination width}, which is based on the well-known notion of \emph{treewidth} (see Section \ref{sec:tractability} for precise definitions). 
For a class $\C$ of well-designed queries, let us denote by $\wdeval(\C)$ the evaluation problem $\wdeval$ restricted to the class $\C$. 
Also, we say that a class $\C$ of well-designed queries has \emph{bounded} domination width if there is an universal constant $k\geq 1$ such that the domination width of every query in $\C$ is at most $k$. 
Then, our main technical result is as follows (Theorem \ref{theo:main}). 
\emph{Assume that FPT $\neq$ W[1]. Then, for every recursively enumerable class $\C$ of 
well-designed queries, the problem $\wdeval(\C)$ is in PTIME if and only if 
$\C$ has bounded domination width}. 
The assumption FPT $\neq$ W[1] is a widely believed assumption from parameterised complexity (see Section \ref{sec:hardness} for precise definitions). 
As we observe in Section \ref{sec:tractability}, one can remove the assumption of $\C$ being recursively enumerable by considering a stronger assumption than FPT $\neq$ W[1]
considering non-uniform complexity classes. 

Our result builds on the classical result by Dalmau et al. \cite{DKV} and Grohe \cite{gro07} showing that a recursively enumerable class of \emph{conjunctive queries} (CQs) over schemas of \emph{bounded arity} is tractable if and only if the \emph{cores} of the CQs in $\C$ have bounded treewidth. 
(Recall that a CQ is a first-order query using only 
conjunctions and existential quantification.)

For the tractability part of our result, we exploit, as in \cite{DKV}, the so-called \emph{existential pebble game} introduced in \cite{KV95} (see also \cite{DKV}). 
This game provides a polynomial-time relaxation for the problem of checking the existence of homomorphisms, 
which is a well-known NP-complete problem (see e.g. \cite{CM77}). 
Using the existential pebble game, we define a natural relaxation of the standard algorithm from \cite{lete} (see also \cite{PS14}) for evaluating well-designed queries. 
Then we show that this relaxation correctly solves instances of bounded domination width (Theorem \ref{theo:main-tractab}). 

For the hardness part, we follow a similar strategy as in \cite{gro07}. The two main ingredients in our proof is 
an adaptation of the main construction of \cite{gro07} to handle \emph{distinguished elements} or \emph{constants} (Lemma \ref{lemma:groheB}) and 
an elementary property of well-designed queries of large domination width (Lemma \ref{lemma:large-dw}). 

Finally, we emphasise that our classes of bounded domination width significantly extend the classes that are locally tractable \cite{lete}, which, as we mentioned above, 
are the most general tractable restrictions known so far.  
This is even true in the case of UNION-free well-designed queries. 
As we discuss in Section \ref{sec:union-free}, the notion of domination width for UNION-free queries can be simplified and coincides with a width measure called \emph{branch treewidth}. 
Bounding this simpler width measure still \emph{strictly} generalises local tractability. 

\medskip
\noindent
{\bf Organisation.} We present  the basic definitions in Section \ref{sec:prel}. In 
Section \ref{sec:tractability}, we introduce the measure of domination width and present our main tractability result. 
The main hardness result is presented in Section \ref{sec:hardness}. 
We conclude with some final remarks in Section \ref{sec:conclusions}. 

\section{Preliminaries}
\label{sec:prel}

\medskip
\noindent
{\bf RDF Graphs.} Let $\bI$ be a countable infinite set of IRIs. 
An \emph{RDF triple} is a tuple in $\bI\times \bI\times \bI$ and 
an \emph{RDF graph} is a finite set of RDF triples. 
In this paper, we assume that no blank nodes appear in RDF graphs, i.e., we focus on \emph{ground} RDF graphs. 

\medskip
\noindent
{\bf SPARQL Syntax.} SPARQL \cite{SPARQL} is the standard query language for RDF. 
We rely on the formalisation proposed in \cite{PAG09}. 
We focus on the core fragment of the language given by the operators $\AND$, OPTIONAL ($\OPT$ for short), and $\UNION$.\footnote{Additional operators include FILTER and SELECT. We briefly discuss these operators in Section \ref{sec:conclusions}.} 
Let $\bV=\{?x,?y,\dots\}$ be a countable infinite set of \emph{variables}, disjoint from $\bI$. 
A SPARQL \emph{triple pattern} (or \emph{triple pattern} for short) is a tuple in $(\bI\cup\bV)\times(\bI\cup\bV)\times (\bI\cup\bV)$. 
The set of variables from $\bV$ appearing in a triple pattern $t$ is denoted by $\vars{t}$. 
Note that an RDF triple is simply a SPARQL triple pattern $t$ with $\vars{t}=\emptyset$. 
A SPARQL \emph{graph pattern} (or \emph{graph pattern} for short) is recursively defined as follows:
\begin{enumerate}
\item a triple pattern is a graph pattern, and 
\item if $P_1$ and $P_2$ are graph patterns, then $P_1\, *\,  P_2$ is also a graph pattern, for $*\in\{\AND,\OPT,\UNION\}$.
\end{enumerate}

\medskip
\noindent
{\bf SPARQL Semantics.} In order to define the semantics of graph patterns, 
we follow again the presentation in \cite{PAG09}. 
A \emph{mapping} $\mu$ is a partial function from $\bV$ to $\bI$. 
We denote by $\dom{\mu}$ the domain of the mapping $\mu$.
Two mappings $\mu_1$ and $\mu_2$ are \emph{compatible} if 
$\mu_1(?x)=\mu_2(?x)$, for all $?x\in \dom{\mu_1}\cap\dom{\mu_2}$. 
If $\mu_1$ and $\mu_2$ are compatible mappings then $\mu_1\cup\mu_2$ denotes the mapping 
with domain $\dom{\mu_1}\cup\dom{\mu_2}$ such that $\mu_1\cup\mu_2(?x)=\mu_1(?x)$, for all $?x\in \dom{\mu_1}$, 
and $\mu_1\cup\mu_2(?x)=\mu_2(?x)$, for all $?x\in \dom{\mu_2}$. 
For a triple pattern $t$ and a mapping $\mu$ such that $\vars{t}\subseteq \dom{\mu}$, 
we denote by $\mu(t)$ the RDF triple obtained from $t$ by replacing each $?x\in \vars{t}$ by $\mu(?x)$. 

For an RDF graph $G$ and a graph pattern $P$, the \emph{evaluation} $\sem{P}{G}$ 
of $P$ over $G$ is a set of mappings defined recursively as follows: 
\begin{enumerate}
\item $\sem{t}{G}=\{\mu\mid \text{$\dom{\mu}=\vars{t}$ and $\mu(t)\in G$}\}$, 
if $t$ is a triple pattern. 
\item $\sem{P_1\, \AND\, P_2}{G}=\{\mu_1\cup\mu_2$ $\mid \mu_1\in \sem{P_1}{G}$, $\mu_2\in \sem{P_2}{G}$ and $\mu_1, \mu_2$ are compatible$\}$. 
\item $\sem{P_1\, \OPT\,  P_2}{G}=\sem{P_1\, \AND\, P_2}{G}$ $\cup$ $\{\mu_1\mid $ $\mu_1\in \sem{P_1}{G}$ and there is no $\mu_2\in \sem{P_2}{G}$ compatible with $\mu_1\}$. 
\item $\sem{P_1\, \UNION\, P_2}{G}=\sem{P_1}{G}\cup\sem{P_2}{G}$.
\end{enumerate}

\medskip
\noindent
{\bf Well-designed SPARQL.} A central class of SPARQL graph patterns identified in \cite{PAG09}, and also the focus of this paper, is the class of 
well-designed graph patterns. 
We say that a graph pattern is \UNION-\emph{free} if it only uses the operators $\AND$ and $\OPT$. 
A \UNION-free graph pattern $P$ is \emph{well-designed} if for every 
subpattern $P'=(P_1\, \OPT\, P_2)$ of $P$, it is the case that every variable $?x$ ocurring in $P_2$ but not in $P_1$, 
does \emph{not} occur outside $P'$ in $P$. 
A SPARQL graph pattern $P$ is \emph{well-designed} if it is of the form $P=P_1\,\UNION\cdots\UNION\, P_m$, where 
each $P_i$ is a \UNION-free well-designed graph pattern.\footnote{This top-level use of the $\UNION$ operator is known as $\UNION$-normal form \cite{PAG09}. 
Note that we are implicitly using the fact that 
$\UNION$ is associative.}

\begin{example}
\label{ex:pattern}
Consider the following graph patterns:
\emph{
\begin{align*}
P_1& =((?x,p,?y)\, \OPT\, (?z,q,?x)) \\
&\qquad\qquad \, \OPT\, ((?y,r,?o_1)\, \AND\, (?o_1,r,?o_2)),\\
P_2& =((?x,p,?y)\, \OPT\, (?z,q,?x)) \\
&\qquad\qquad \, \OPT\, ((?y,r,?z)\, \AND\, (?z,r,?o_2)).
\end{align*}}
Note that $P_1$ is well-designed, while $P_2$ is not. 
Indeed, in the subpattern \emph{$P_2'=((?x,p,?y)\, \OPT\, (?z,q,?x))$} of $P_2$, the variable $?z$ appears in $(?z,q,?x)$ and not in $(?x,p,?y)$ but \emph{does} occur outside $P'_2$ in $P_2$.
\end{example}

Well-designed patterns have good properties in terms of query evaluation. 
More precisely, 
let $\wdeval$ be the problem of deciding, given a well-designed graph pattern $P$, an RDF graph $G$ and a mapping $\mu$, 
whether $\mu\in \sem{P}{G}$. 
It was shown in \cite{PAG09} that $\wdeval$ is $\coNP$-complete, 
while the problem is $\PSPACE$-complete for arbitrary SPARQL graph patterns.

\subsection{Pattern trees and pattern forests}

Besides alleviating the cost of evaluation, another key property of 
$\UNION$-free well-designed graph patterns is that they can be written in the so-called 
$\OPT$-normal form \cite{PAG09}. 
In turn, patterns in $\OPT$-normal form admit a natural tree representation, 
known as \emph{pattern trees} \cite{lete}. 
Intuitively, a pattern tree is a rooted tree where each node represents a well-designed pattern using only $\AND$ operators, while 
its tree structure  represents the nesting of $\OPT$ operators. 
Consequently, a well-designed graph pattern $P=$ $P_1\,$ $\UNION$ $\cdots$ $\UNION\,$ $P_m$ can be represented as 
a \emph{pattern forest}\footnote{In this paper, we work with a particular type of patterns trees/forests, namely \emph{well-designed pattern trees/forests}. For simplicity, sometimes we abuse notation and use the terms patterns trees/forests and well-designed pattern trees/forests interchangeably.}\cite{PS14}, i.e., a set of pattern trees $\{\T_1,\dots,\T_m\}$, where $\T_i$ is the pattern tree representation of $P_i$. 
Pattern trees/forests are useful for understanding how to evaluate and optimise well-designed patterns, 
and have been used extensively as a basic tool in the study of well-designed SPARQL (see e.g. \cite{lete,PS14,BPS15,KRRV15,KK16,KPS16}). 
As we show in this work, pattern forests are also fundamental to understand tractable evaluation of well-designed SPARQL: 
by imposing restrictions on the pattern forest representation, 
we can identify and \emph{characterise} the tractable classes of well-designed 
graph patterns. 

\medskip
\noindent
{\bf T-graphs and homomorphisms.} A \emph{triple pattern graph} (or \emph{t-graph} for short) is a finite set $S$ of triple patterns. 
We denote by $\vars{S}$ the set of variables from $\bV$ appearing in the t-graph $S$. 
Note that an RDF graph is simply a t-graph $S$ with $\vars{S}=\emptyset$. 
Let $t$ be a triple pattern and $h$ be a partial function from $\bV$ to $\bI\cup\bV$ such that $\vars{t}\subseteq\dom{h}$. 
We define $h(t)$ to be the triple pattern obtained from $t$ by replacing each $?x\in \vars{t}$ by $h(?x)$. 
For two t-graphs $S$ and $S'$, we say that a partial function $h$ from $\bV$ to $\bI\cup\bV$ is a $\emph{homomorphism}$ 
from $S$ to $S'$ if $\dom{h}=\vars{S}$ and for every $t\in S$, it is the case that $h(t)\in S'$. 

\medskip
\noindent
{\bf Basics of pattern trees and forests.} For an undirected graph $H$, we denote by $V(H)$ its set of nodes. 
A \emph{well-designed pattern tree} (or wdPT for short) 
is a triple $\T=(T,r,\lambda)$ such that 
\begin{enumerate}
\item $T$ is a tree rooted at a node $r\in V(T)$, 
\item $\lambda$ is a function that maps each node $n\in V(T)$ to a t-graph, and  
\item the set $\{n\in V(T)\mid?x\in \vars{\lambda(n)}\}$ induces a connected subgraph of $T$, 
for every $?x\in \bV$. 
\end{enumerate}

Let $\T=(T,r,\lambda)$ be a wdPT. 
A wdPT $\T'=(T',r',\lambda')$ is a \emph{subtree} of $\T$ if (i) $T'$ is a subtree of $T$, (ii) $r'=r$, 
and $\lambda'(n)=\lambda(n)$, for all $n\in V(T')$.  
Note that any subtree of $\T$ contains the original root $r$.
A \emph{child} of the subtree $\T'$ is a node $n\in V(T)\setminus V(T')$ such that $n'\in V(T')$, 
where $n'$ is the parent of $n$ in $T$.

For convenience, we fix two functions pat($\cdot$) and vars($\cdot$) as follows. 
Let $\T=(T,r,\lambda)$ be a wdPT. 
We define $\pat{n}:=\lambda(n)$, for every $n\in V(T)$ and $\pat{\T}:=\bigcup_{n\in V(T)} \pat{n}$. 
Note that $\pat{n}$ and $\pat{\T}$ are t-graphs. 
We let $\vars{n}:=\vars{\pat{n}}$, for $n\in V(T)$ and $\vars{\T}:=\vars{\pat{\T}}$.

A \emph{well-designed pattern forest} (wdPF for short) is a finite set $\F=\{\T_1,\dots,\T_m\}$ of well-designed pattern trees.

In \cite{lete}, it was shown that every wdPT can be translated efficiently into an 
equivalent wdPT in the so-called NR normal form. 
A wdPT $\T=(T,r,\lambda)$ is in \emph{NR normal form} if for every node $n\in V(T)$ 
with parent $n'$ in $T$, it holds that $\vars{n}\setminus \vars{n'}\neq \emptyset$. 
In this paper, we assume that all wdPTs are in NR normal form.

\medskip
\noindent
{\bf Well-designed SPARQL and wdPFs.} As in the case of SPARQL graph patterns, 
we denote by $\sem{\T}{G}$ (resp., $\sem{\F}{G}$) the evaluation of a wdPT $\T$ (resp., wdPF $\F$) over an RDF graph $G$. 
In \cite{lete}, for a wdPT $\T$, the set of mappings $\sem{\T}{G}$ 
is defined via a translation to well-designed graph patterns. 
However, if $\T$ is in NR-normal form, then $\sem{\T}{G}$ 
admits a simple characterisation stated in Lemma \ref{lemma:char-pt} below. 
In this paper, we adopt this characterisation as the semantics of wdPTs.

\begin{lemma}[\cite{lete,PS14}]
\label{lemma:char-pt}
Let $\T$ be a wdPT in NR normal form, $G$ an RDF graph and $\mu$ a mapping. 
Then $\mu\in \sem{\T}{G}$ iff there exists a subtree $\T'$ of $\T$ 
such that 
\begin{enumerate}
\item $\mu$ is a homomorphism from $pat(\T')$ to $G$. 
\item there is no child $n$ of $\T'$ and homomorphism $\nu$ from $pat(n)$ to $G$ 
compatible with $\mu$. 
\end{enumerate}
\end{lemma}

For a wdPF $\F=\{\T_1,\dots,\T_m\}$ and an RDF graph $G$, we define $\sem{\F}{G}=\sem{\T_1}{G}\cup\cdots\cup\sem{\T_m}{G}$. 

\medskip
As shown in \cite{lete}, every \UNION-free well-designed graph pattern $P$
can be translated in polynomial time into an equivalent wdPT $\T$, i.e., 
a wdPT such that $\sem{\T}{G}=\sem{P}{G}$, for all RDF graphs $G$. 
Consequently and as observed in \cite{PS14}, every well-designed graph pattern $P$
can be translated in polynomial time into an equivalent wdPF $\F$. 
\emph{Throughout the paper, we fix a polynomial-time computable function ${\bf wdpf}$
that maps each well-designed graph pattern to an equivalent wdPF. }

\begin{example}
Recall $P_1$ from Example \ref{ex:pattern} and consider the following well-designed graph pattern:
\emph{
\begin{align*}
P& =P_1 \, \UNION\, \\
&\qquad\qquad((?x,p,?y)\, \OPT\, ((?z,q,?x)\, \AND\, (?w,q,?z))).
\end{align*}}
We have that $\wdpf(P)=\{\T_1,\T_2\}$, where $\T_1$ and $\T_2$ are the wdPTs depicted in Figure \ref{fig:main}, 
for $k=2$ and $K_2(?o_1,?o_2)=\{(?o_1,r,?o_2)\}$. 
\end{example}

\subsection{Restrictions of the evaluation problem}

Recall that $\wdeval$ denotes the problem of deciding, given a well-designed graph pattern $P$, an RDF graph $G$ and a mapping $\mu$, 
whether $\mu\in \sem{P}{G}$. 
In this paper, we study restrictions of $\wdeval$ given by different classes $\C$ of admissible patterns. 
Formally, for a class $\C$ of well-designed graph patterns, we define the problem $\wdeval(\C)$ as follows: 

\begin{center}
\begin{tabular}{|l|}
\hline
$\wdeval(\C)$\\
{\bf Input}: a well-designed graph pattern $P\in \C$,\\
an RDF graph $G$ and a mapping $\mu$. \\
{\bf Question}: does $\mu\in \sem{P}{G}$ hold?\\
\hline
\end{tabular}
\end{center}

Note that $\wdeval(\C)$ is a \emph{promise} problem, as 
we are given the promise that $P\in \C$. 
This allows us to analyse the complexity of evaluating patterns in $\C$ \emph{independently} of the cost of  
checking membership in $\C$.

\section{A new tractability condition} 
\label{sec:tractability}

In this section, we introduce the notion of domination width of a well-designed graph pattern and show our main tractability result: 
$\wdeval(\C)$ is in $\ptime$, for classes $\C$ of graph patterns of bounded domination width. 
Before doing so, we need to introduce some terminology.  

A \emph{generalised t-graph} is a pair $(S,X)$, where $S$ is a t-graph and $X\subseteq \vars{S}$. 
Consider two generalised t-graphs of the form $(S,X)$ and $(S',X)$. A \emph{homomorphism} from $(S,X)$ to $(S',X)$ 
is a homomorphism $h$ from $S$ to $S'$ such that $h(?x)=?x$, for all $?x\in X$. 
We write $(S,X)\rightarrow (S',X)$ whenever there is a homomorphism from $(S,X)$ to $(S',X)$; 
otherwise, we write $(S,X)\not\rightarrow (S',X)$. Note that the relation $\rightarrow$ is transitive, i.e., 
$(S,X)\rightarrow (S',X)$ and $(S',X)\rightarrow (S'',X)$ implies $(S,X)\rightarrow (S'',X)$.

Let $(S,X)$ be a generalised t-graph, $G$ be an RDF graph and $\mu$ be a mapping with $\dom{\mu}=X$. 
We write $(S,X)\rightarrow^\mu G$ if there is a homomorphism $h$ from $S$ to $G$ such that $h(?x)=\mu(?x)$, 
for all $?x\in X$. Notice that $\rightarrow$ composes with $\rightarrow^\mu$, i.e., 
$(S,X)\rightarrow (S',X)$ and $(S',X)\rightarrow^\mu G$ implies $(S,X)\rightarrow^\mu G$. 

Below we state several notions and properties for generalised t-graphs. 
We emphasise that all these properties are well-known for \emph{conjunctive queries} (CQs) and \emph{relational structures} 
and can be applied in our case as there is a strong correspondence between generalised t-graphs and CQs. 
Indeed, we can view a generalised t-graph $(S,X)$ as a CQ $q_{(S,X)}$ over a relational schema containing a single ternary relation, 
where the variables are $\vars{S}$, the \emph{free variables} are $X$, and the IRIs appearing in $S$ correspond to \emph{constants} in $q_{(S,X)}$. 
However, for convenience and consistency with RDF and SPARQL terminology, we shall work directly with  generalised t-graphs throughout the paper. 

\medskip
\noindent
{\bf Cores.} Let $(S,X)$ and $(S',X)$ be two generalised t-graphs. 
We say that $(S',X)$ is a \emph{subgraph} of $(S,X)$ if $S'\subseteq S$, and a \emph{proper} subgraph if $S'\subseteq S$ but $S\not\subseteq S'$. 
 A generalised t-graph $(S,X)$ is a \emph{core} if there is no homomorphism from $(S,X)$ to one of its proper subgraphs $(S',X)$. 
 We say that $(S',X)$ is a \emph{core} of $(S,X)$ if $(S',X)$ is a core itself, $(S,X)\rightarrow (S',X)$ and  $(S',X)\rightarrow (S,X)$. 
As stated below, every generalised t-graph $(S,X)$ has a unique core (up to renaming of variables), and hence, we can 
 speak of \emph{the} core of a generalised t-graph. 
  
\begin{proposition}[see e.g. \cite{AHV95,HN92}] 
\label{prop:core}
 Every  generalised t-graph $(S,X)$ has a unique core $(S',X)$ (up to renaming of variables).
 \end{proposition}

\medskip
\noindent
{\bf Treewidth.} The notion of treewidth is a well-known measure of the tree-likeness of an undirected graph (see e.g. \cite{diestel}).  
For instance, trees have treewidth $1$, cycles treewidth $2$ and $K_k$, the clique of size $k$, treewidth $k-1$. 
Let $H$ be an undirected graph. A \emph{tree decomposition} of $H$ 
is a pair $(F,\beta)$ where $F$ is a tree and $\beta$ is a function that maps 
each node $s\in V(F)$ to a subset of $V(H)$ such that
\begin{enumerate}
\item for every $u\in V(H)$, the set 
$\{s\in V(F)\mid u\in \beta(s)\}$ induces a connected subgraph of $F$, and 
\item for every edge $\{u,v\}\in E(H)$, there is a node $s\in V(F)$ with $\{u,v\}\subseteq \beta(s)$. 
\end{enumerate} 
The \emph{width} of the decomposition $(F,\beta)$ is $\max\{|\beta(s)|\mid s\in V(F)\}-1$. 
The \emph{treewidth} $\tw{H}$ of the graph $H$ is the minimum width over all its tree decompositions. 

Let $(S,X)$ be a generalised t-graph. The \emph{Gaifman graph} $G(S,X)$ of $(S,X)$ is the undirected graph 
whose vertex set is $\vars{S}\setminus X$ and whose edge set contains the pairs $\{?x,?y\}$ such that $?x\neq ?y$ and 
$\{?x,?y\}\subseteq\vars{t}$, for some triple pattern $t\in S$. 
We define the \emph{treewidth} of $(S,X)$ to be $\tw{S,X}:=\tw{G(S,X)}$. 
If $G(S,X)$ has no vertices, i.e., $\vars{S}\setminus X=\emptyset$, or $G(S,X)$ has no edges, we let 
$\tw{S,X}=\tw{G(S,X)}:=1$.

For a generalised t-graph $(S,X)$, we let $\ctw{S,X}:=\tw{S',X}$, 
where $(S',X)$ is the core of $(S,X)$.
 
 \begin{figure}
\centering
\includegraphics[scale=0.6]{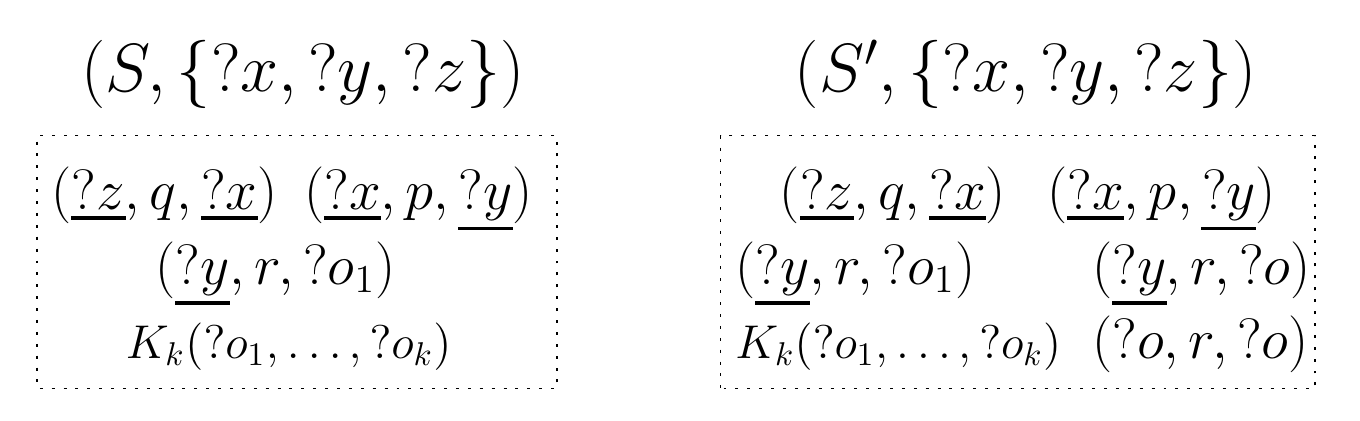}
\caption{\normalfont{The generalised t-graphs from Example \ref{ex:cores}. We assume that $k\geq 2$ and $K_k(?o_1,\dots,?o_k)=\{(?o_i,r,?o_j)\mid i,j\in\{1,\dots,k\} \text{ with $i<j$}\}$. Note that the distinguished variables are underlined.}}
\label{fig:tw}
\end{figure}
 
 \begin{example}
 \label{ex:cores}
 Let $X=\{?x,?y,?z\}$ and consider the generalised t-graphs $(S,X)$ and $(S',X)$ depicted in Figure \ref{fig:tw}, 
 where $k\geq 2$ and $K_k(?o_1,\dots,?o_k)$ is the t-graph given by the set 
 \begin{align*}
  K_k(?o_1,\dots,&?o_k):=\\
  &\{(?o_i,r,?o_j)\mid i,j\in \{1,\dots,k\} \text{ with $i<j$}\}.
 \end{align*}
 Observe that $(S,X)$ is a core and hence $\ctw{S,X}=k-1$, as its Gaifman graph is the clique of size $k$. 
 On the other hand, the core of $(S',X)$ is $(C',X)$, where 
 $$C'=\{(?z,q,?x), (?x,p,?y), (?y,r,?o), (?o,r,?o)\}.$$ 
 Hence, $\ctw{S',X}=1$ while $\tw{S',X}=k-1$. 
 \end{example}

\medskip
\noindent
{\bf Existential $k$-pebble game.} The \emph{existential $k$-pebble game} was introduced by Kolaitis and Vardi \cite{KV95} 
to analyse the expressive power of certain \emph{Datalog programs}. 
While the original definition deals with relational structures, 
here we focus on the natural adaptation to the context of generalised t-graphs and RDF graphs. 

Let $k \geq 2$. 
The existential $k$-pebble game is played by the \emph{Spoiler} and the \emph{Duplicator} 
on a generalised t-graph $(S,X)$, an RDF graph $G$ and a mapping $\mu$ with $\dom{\mu}=X$. 
During the game, the Spoiler only picks elements from $\vars{S}\setminus X$, while the Duplicator picks elements from $\dom{G}$,
where $\dom{G}\subseteq \bI$ is the set of IRIs appearing in $G$. 
In the first round, the Spoiler places pebbles on (not necessarily distinct) elements $?x_1,\dots,?x_k \in \vars{S}\setminus X$, 
and the Duplicator responds by placing pebbles on elements $a_1,\dots,a_k\in\dom{G}$. 
On any further round, the Spoiler removes a pebble and places it on another element $?x\in\vars{S}\setminus X$. 
The Duplicator responds by moving the corresponding pebble to an element $a\in\dom{G}$. 
If after a particular round, the elements covered by the pebbles are $?x_1,\dots,?x_k$ and $a_1,\dots,a_k$ 
for the Spoiler and the Duplicator, respectively, then the \emph{configuration} of the game 
is $\perp$ if $?x_i=?x_j$ and $a_i\neq a_j$, for some $i,j\in\{1,\dots,k\}$ with $i\neq j$; 
otherwise, it is the mapping $\mu\cup\nu$, where $\dom{\nu}=\{?x_1,\dots,?x_k\}$ and $\nu(?x_i)=a_i$, for every $i\in\{1,\dots,k\}$ 
(note that $\dom{\mu}\cap\dom{\nu}=\emptyset$).

The Duplicator wins the game if he has a \emph{winning strategy}, that is, 
he can indefinitely continue playing the game in such a way that the configuration at the end of each round is a mapping $\mu\cup\nu$ that is a 
\emph{partial homomorphism}, i.e., for every triple pattern $t\in S$ with $\vars{t}\subseteq \dom{\mu\cup\nu}$, 
it is the case that $\mu\cup\nu(t)\in G$.  
If the Duplicator can win the existential $k$-pebble game on $(S,X)$, $G$ and $\mu$, then we write 
$(S,X)\rightarrow^\mu_k G$.

Note that if $\vars{S}\setminus X=\emptyset$, then for every $k\geq 2$, 
\begin{equation}
\label{eq:empty}
(S,X)\rightarrow^\mu_k G \text{ if and only if } (S,X)\rightarrow^\mu G, 
\end{equation} 
i.e., $\mu$ is a homomorphism from $S$ to $G$. 
Observe also that for every $k\geq 2$, 
\begin{equation}
\label{eq:relax}
(S,X)\rightarrow^\mu G\text{ implies }(S,X)\rightarrow^\mu_k G. 
\end{equation}
In other words, the relation $\rightarrow_k^\mu$ is a \emph{relaxation} of $\rightarrow^\mu$. 
As we state below, the relaxation given by $\rightarrow_k^\mu$ has good properties in terms of complexity\footnote{The existential $k$-pebble game is known to capture the so-called \emph{$k$-consistency test} \cite{KV00b}, 
which is a well-known heuristic for solving \emph{constraint satisfaction problems} (CSPs).}: 
while checking  the existence of homomorphisms, i.e., $(S,X)\rightarrow^\mu G$ is a well-known NP-complete problem \cite{CM77}, 
checking $(S,X)\rightarrow^\mu_k G$ can be done in polynomial time, for every fixed $k\geq 2$. 

\begin{proposition}[\cite{KV95}; see also \cite{DKV}]
\label{prop:games-ptime}
Let $k\geq 2$. For a given generalised t-graph $(S,X)$, an RDF graph $G$ and a mapping $\mu$ with $\dom{\mu}=X$, 
checking whether $(S,X)\rightarrow^\mu_k G$ can be done in polynomial time. 
\end{proposition} 

As it turns out, there is a strong connection between existential $k$-pebble games and the notion of treewidth. 
In particular, it was shown by Dalmau et al. \cite{DKV} that the relations $\rightarrow_k$ and $\rightarrow$ coincide 
for generalised t-graphs $(S,X)$ satisfying $\ctw{S,X}\leq k-1$\footnote{In \cite{DKV}, it was shown that $\rightarrow_k$ and $\rightarrow$ coincide for relational structures whose cores have treewidth at most $k-1$. For Proposition \ref{prop:games-tw}, we need a generalisation of the results in \cite{DKV} that considers relational structures equipped with a set of \emph{distinguished elements}. Indeed, 
such distinguished elements correspond to the variables in $X$ and the IRIs appearing in the generalised t-graph $(S,X)$. Such a generalisation follows straightforwardly from the results in \cite{DKV}.}.

\begin{proposition}[\cite{DKV}]
\label{prop:games-tw}
Let $k\geq 2$. Let $(S,X)$ be a generalised t-graph, $G$ be an RDF graph and $\mu$ be a mapping with $\dom{\mu}=X$. 
Suppose that $\ctw{S,X}\leq k-1$. Then $(S,X)\rightarrow^\mu_k G$ if and only if $(S,X)\rightarrow^\mu G$. 
\end{proposition} 

We conclude  with two basic properties of the existential pebble game that will be useful for us. 

\begin{proposition}
\label{prop:games}
Let $k\geq 2$. Let $(S_1,X)$, $(S_2,X)$, $\dots$, $(S_\ell,X)$ be generalised t-graphs $(\ell \geq 2)$, $G$ be an RDF graph and $\mu$ be a mapping with $\dom{\mu}=X$. 
Then the following hold:
\begin{enumerate}
\item if $(S_1,X)\rightarrow (S_2,X)$ and $(S_2,X)\rightarrow^\mu_k G$, then it is the case that  $(S_1,X)\rightarrow^\mu_k G$. 
\item if $(S_i,X)\rightarrow^\mu_k G$, for all $i\in\{1,\dots,\ell\}$ and $(\vars{S_i}\setminus X)\cap (\vars{S_j}\setminus X)=\emptyset$, 
for all $i,j\in\{1,\dots,\ell\}$ with $i\neq j$, then $(S_1\cup\cdots\cup S_\ell, X)\rightarrow^\mu_k G$. 

\end{enumerate}
\end{proposition}

\subsection{Domination width}

We start by giving some intuition regarding the notion of domination width. 
Let $P$ be a well-designed graph pattern, $G$ be an RDF graph and $\mu$ be a mapping. 
Suppose that $\wdpf(P)=\F$ and $\F=\{\T_1,\dots,\T_m\}$, for $m\geq 1$. 
The natural algorithm for checking $\mu\not\in\sem{\F}{G}$ is as follows (see e.g. \cite{lete,PS14}): 
we simply iterate over all $i\in\{1,\dots,m\}$ such that $\mu$ is a \emph{potential solution} of $\T_i$ over $G$, i.e., there is a subtree $\T'_i$ of $\T_i$ such that $\mu$ is a homomorphism 
from $\pat{\T'_i}$ to $G$, and we ensure that there is a child $n_i$ of $\T'_i$ where $\mu$ can be extended consistently. 

 The key observation is that we can reinterpret the above-described algorithm as follows. 
We can choose one of the subtrees $\T'_i$ as above, and associate a collection of generalised t-graphs $\GtG{\T'_i}$ of 
 the form $(S,\vars{\T'_i})$, where $S=\pat{\T'_i}\cup\bigcup_{j\in I} \pat{n_j}$, where $I\subseteq \{1,\dots,m\}$ is the set of indices $j$ such that $\mu$ is a potential solution of $\T_j$ over $G$, 
 and $n_j$ is a child of $\T'_j$. To avoid conflicts, for every $j\in I$, the variables from $\vars{n_j}$ that are not in $\vars{\T'_i}=\dom{\mu}$, need to be renamed to fresh variables. 
 Therefore, checking  $\mu\not\in\sem{\F}{G}$ amounts to checking that there is a homomorphism from some element of $\GtG{\T'_i}$ to $G$, i.e., 
 whether $(S,\vars{\T'_i})\rightarrow^\mu G$, for some $(S,\vars{\T'_i})\in \GtG{\T'_i}$. 
 
 The idea behind domination width is to ensure that $\GtG{\T'_i}$ is always \emph{dominated} by a subset $\G\subseteq \GtG{\T'_i}$ where each generalised t-graph in $\G$ has small $\text{ctw}$. 
 The set $\G$ dominates $\GtG{\T'_i}$ in the sense that, for every $(S',\vars{\T'_i})\in\GtG{\T'_i}$, there is a $(S,\vars{\T'_i})\in\G$ such that $(S,\vars{\T'_i})\rightarrow (S',\vars{\T'_i})$. 
 Therefore, by transitivity of the relation $\rightarrow$, checking $\mu\not\in\sem{\F}{G}$ amounts to checking that there is a homomorphism from some element of $\G$ to $G$. 
 Since generalised t-graphs of small ctw are well-behaved with respect to the relaxation $\rightarrow_k$ (see Proposition \ref{prop:games-tw}), 
 this will imply that the relaxation of the natural algorithm, described at the beginning of this section, given by replacing homomorphism tests $\rightarrow$ by $\rightarrow_k$,  
 correctly decides if $\mu\not\in\sem{\F}{G}$. 
 Below we formalise this intuition. 

Let $\F=\{\T_1,\dots,\T_m\}$ be a wdPF. 
A \emph{subtree} $\T$ of $\F$ is a subtree of some wdPT $\T_i$, for $i\in\{1,\dots,m\}$. 
The \emph{support} $\supp{\T}$ of the subtree $\T$ contains precisely the indices $i$ from
$\{1,\dots,m\}$ such that there is a subtree $\T_i'$ of $\T_i$ satisfying $\vars{\T_i'}=\vars{\T}$.   
Note that $\supp{\T}\neq \emptyset$, for every subtree $\T$. 
Since wdPTs are in NR normal form, whenever $i\in \supp{\T}$, then the witness subtree $\T'_i$ is unique. 
For $i\in\supp{\T}$, we denote such a $\T'_i$ by $\T^{\tsupp}(i)$.

Let $\T$ be a subtree of $\F=\{\T_1,\dots,\T_m\}$. 
A \emph{children assignment} for $\T$ is a function $\Delta$ with a non-empty domain $\dom{\Delta}\subseteq \supp{\T}$ that maps every $i\in\dom{\Delta}$ to a child $\Delta(i)$ of $\T^{\tsupp}(i)$. 
We denote by $\CA{\T}$ the set of all children assignments for $\T$. 
Observe that if $\Delta\in\CA{\T}$, then it must be the case that $\T^{\tsupp}(i)\neq \T_i$, for every $i\in\dom{\Delta}$. 
In particular, it could be the case that $\CA{\T}=\emptyset$. 
The \emph{renamed t-graphs assignment} $\rho_\Delta$ associated with $\Delta$ maps 
$i\in \dom{\Delta}$ to a t-graph $\rho_\Delta(i)$ obtained from $\pat{\Delta(i)}$ by renaming all variables in $\vars{\Delta(i)}\setminus \vars{\T}$ to new fresh variables. 
In particular, if $i,j\in \dom{\Delta}$ and $i\neq j$, then 
$$(\vars{\rho_\Delta(i)}\setminus\vars{\T}) \cap (\vars{\rho_\Delta(j)}\setminus \vars{\T}) = \emptyset.$$

For $\Delta\in\CA{\T}$, we define the t-graph $S_\Delta$ as 
$$S_\Delta:=\pat{\T}\cup\bigcup_{i\in\dom{\Delta}}\rho_\Delta(i).$$ 
We say that a children assignment $\Delta\in\CA{\T}$ is \emph{valid} if for every $i\in\supp{\T}\setminus \dom{\Delta}$, we have that
$$(\pat{\T^{\tsupp}(i)}, \vars{\T})\not\rightarrow (S_\Delta,\vars{\T}).$$
We denote by $\VCA{\T}$ the set of valid children assignments for $\T$. 
Finally, for the subtree $\T$, we define the set of generalised t-graphs \emph{associated with} $\T$ as
$$\GtG{\T}:=\{(S_\Delta, \vars{\T})\mid \Delta\in \VCA{\T}\}.$$

\begin{figure}
\centering
\includegraphics[scale=0.52]{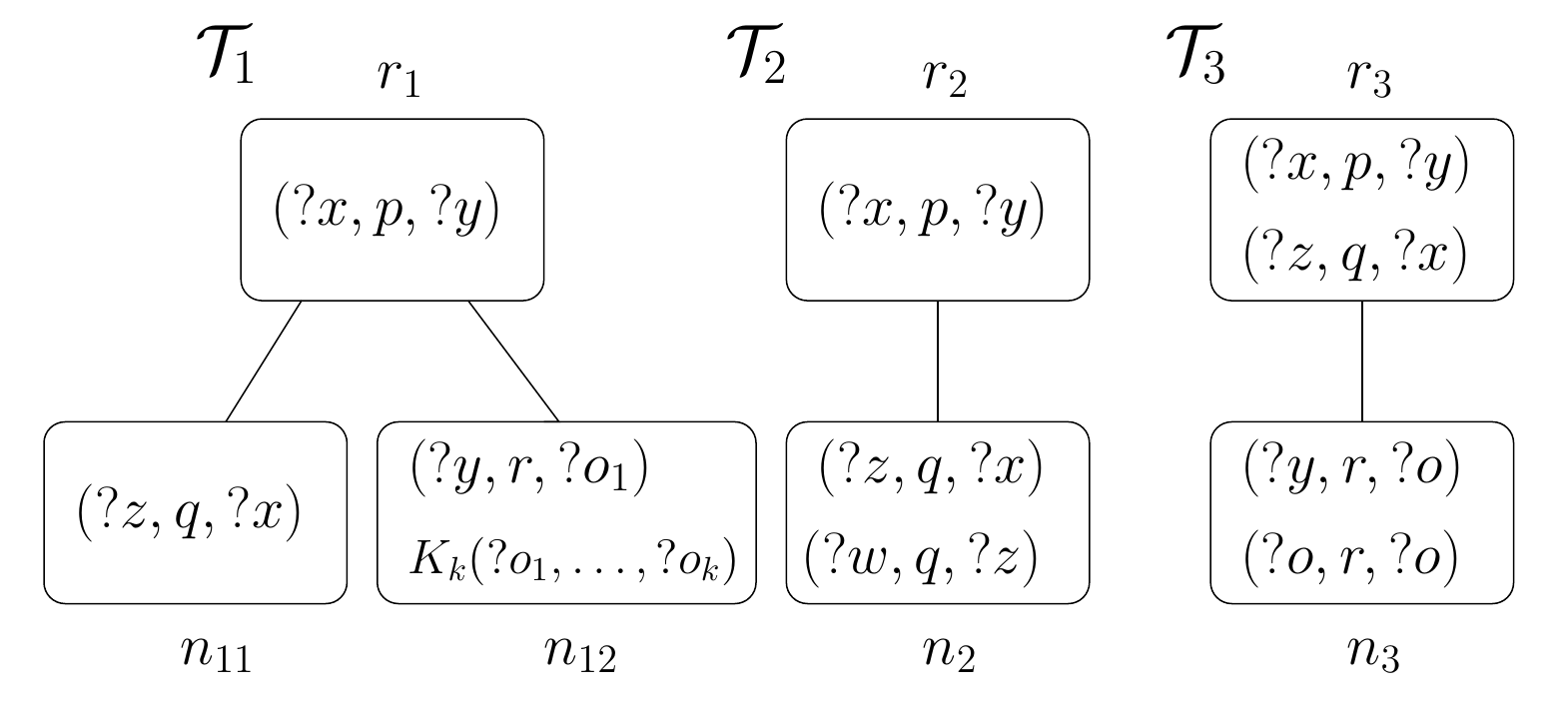}
\caption{\normalfont{The wdPF $\F_k=\{\T_1,\T_2,\T_3\}$ of Example \ref{ex:domination}.}}
\label{fig:main}
\end{figure}

\begin{figure}
\centering
\includegraphics[scale=0.6]{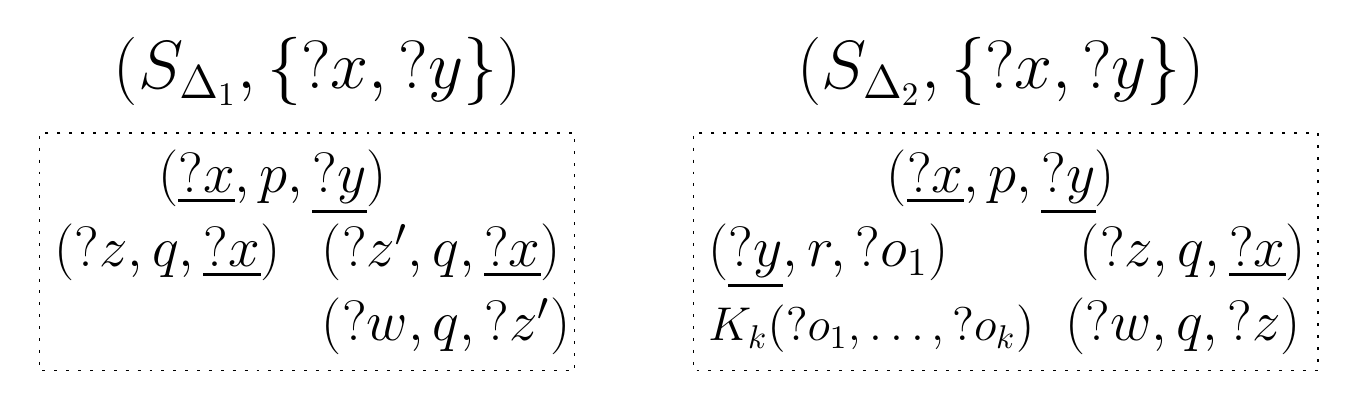}
\caption{\normalfont{The generalised t-graphs $(S_{\Delta_1},\{?x,?y\})$ and $(S_{\Delta_2},\{?x,?y\})$ from Examples \ref{ex:domination} and \ref{ex:dom1}.}} 
\label{fig:dom}
\end{figure}

\begin{example}
\label{ex:domination}
Let $k\geq 2$. Recall from  Example \ref{ex:cores} that 
 \begin{align*}
  K_k(?o_1,\dots,&?o_k):=\\
  &\{(?o_i,r,?o_j)\mid i,j\in \{1,\dots,k\} \text{ with $i<j$}\}.
 \end{align*}
Consider the wdPF $\F_k=\{\T_1,\T_2,\T_3\}$ depicted in Figure \ref{fig:main}. 
For a wdPT $\T=(T,r,\lambda)$ and a subset $N\subseteq V(T)$, we denote by $\T[N]$ 
the subtree of $\T$ induced by the set of nodes $N$. Observe that the only subtrees $\T$ of $\F$ with a non-empty set $\GtG{\T}$ are $\T_1[r_1]$, $\T_1[r_1, n_{11}]$, $\T_1[r_1, n_{12}]$, 
$\T_2[r_2]$ and $\T_3[r_3]$. Consider first $\T_1[r_1]$ and note that $\supp{\T_1[r_1]}=\{1,2\}$. 
We have that 
$$\GtG{\T_1[r_1]}=\{(S_{\Delta_1},\{?x,?y\}), (S_{\Delta_2},\{?x,?y\})\}$$
with $\Delta_1,\Delta_2\in \VCA{\T_1[r_1]}$, where $\Delta_1$ and $\Delta_2$ are described by $\Delta_1=\{1\mapsto n_{11}, 2\mapsto n_2\}$ and $\Delta_2=\{1\mapsto n_{12}, 2\mapsto n_2\}$. 
Figure \ref{fig:dom} illustrates $(S_{\Delta_1},\{?x,?y\})$ and $(S_{\Delta_2},\{?x,?y\})$. Note how we need to rename $?z$ to a fresh variable $?z'$ in $(S_{\Delta_1},\{?x,?y\})$. 
Observe also that, for instance, the children assignment given by $\Delta_3=\{1\mapsto n_{11}\}$ is not valid as $2\not\in\dom{\Delta_3}$ and 
$$(\pat{\T_2[r_2]},\{?x,?y\})\rightarrow(S_{\Delta_3},\{?x,?y\}).$$
For $\T_1[r_1, n_{11}]$, we have that 
$$\GtG{\T_1[r_1, n_{11}]}=\{(S_\Delta,\{?x,?y,?z\})\}$$
where $\Delta=\{1\mapsto n_{12}, 3\mapsto n_3\}$. 
Note that $(S',\{?x,?y,?z\})$ in Figure \ref{fig:tw} corresponds to $(S_\Delta,\{?x,?y,?z\})$. 
In the case of $\T_1[r_1, n_{12}]$, we have that 
$$\GtG{\T_1[r_1, n_{12}]}=\{(S_{\Delta'},\{?x,?y,?o_1,\dots,?o_k\})\}$$
where $\Delta'=\{1\mapsto n_{11}\}$. 
Finally, note that $\GtG{\T_2[r_2]}$ $=$ $\GtG{\T_1[r_1]}$ and $\GtG{\T_3[r_3]}$ $=$ $\GtG{\T_1[r_1,n_{11}]}$.
\end{example}

Now we are ready to define domination width. 

\begin{definition}[$k$-domination]
Let $\G$ be a set of generalised t-graphs of the form $\G=\{(S,X)\mid S\in\ES\}$, 
where $\ES$ is a set of t-graphs and $X$ is a fixed set of variables with $X\subseteq \vars{S}$, for all $S\in\ES$. 
We say that $\G'\subseteq \G$ is a \emph{dominating set} of $\G$ if for every $(S,X)\in\G\setminus \G'$, 
there exists $(S',X)\in \G'$ such that $(S',X)\rightarrow (S,X)$. 

We say that $\G$ is \emph{$k$-dominated} if the set $\{(S,X)\in\G\mid \ctw{S,X}\leq k\}$ is a dominating set of $\G$.
\end{definition}

\begin{definition}[Domination width]
Let $\F$ be a wdPF. The \emph{domination width} of $\F$, denoted by $\dw{\F}$, is the minimum positive integer such that 
for every subtree $\T$ of $\F$, the set of generalised t-graphs $\GtG{\T}$ is $k$-dominated. 

For a well-designed graph pattern $P$, we define the \emph{domination width} of $P$ as $\dw{P}:=\dw{\wdpf(P)}$. 
\end{definition}

We say that a class $\C$ of well-designed graph patterns has \emph{bounded} domination width if there is a universal constant $k\geq 1$ such that 
$\dw{P}\leq k$, for every $P\in\C$. 

\begin{example}
\label{ex:dom1}
Consider a class $\C=\{P_k\mid k\geq 2\}$ such that $\wdpf(P_k)=\F_k$, where $\F_k$ is the wdPF defined in Figure \ref{fig:main} and Example \ref{ex:domination}. 
We claim that $\C$ has bounded domination width as for every $k\geq 2$, it is the case that $\dw{\F_k}=1$. 
Indeed, following the notation from Example \ref{ex:domination}, we need to check that $\GtG{\T_1[r_1]}$, $\GtG{\T_1[r_1, n_{11}]}$ 
and $\GtG{\T_1[r_1, n_{12}]}$ are $1$-dominated. 

Note first that $\ctw{S_{\Delta'},\{?x,?y,?o_1,\dots,?o_k\}}=1$. Therefore, 
$\GtG{\T_1[r_1, n_{12}]}$ is $1$-dominated. Observe also that $(S_\Delta,\{?x,?y,?z\})$ coincides with $(S',\{?x,?y,?z\})$ from Figure \ref{fig:tw} and,  
as explained in Example \ref{ex:cores}, we have $\ctw{S',\{?x,?y,?z\}}=1$. It follows that $\GtG{\T_1[r_1, n_{11}]}$  is also $1$-dominated. 
Finally, for $\T_1[r_1]$, we have that $\ctw{S_{\Delta_1},\{?x,?y\}}$ $=$ $1$ and $\ctw{S_{\Delta_2},\{?x,?y\}}$ $=$ $k-1$ (see Figure \ref{fig:dom}). 
However, we have that $(S_{\Delta_1},\{?x,?y\})$ $\rightarrow$ $(S_{\Delta_2},\{?x,?y\})$, and hence, $\GtG{\T_1[r_1]}$ is also $1$-dominated. 
\end{example}

The following is our main tractability result. 

\begin{theorem}[Main tractability]
\label{theo:main-tractab}
Let $\C$ be a class of well-designed graph patterns of bounded domination width. 
Then $\wdeval(\C)$ is in $\ptime$.  
\end{theorem}

\begin{proof}
Let $k\geq 1$ be a positive integer such that $\dw{P}\leq k$, for all $P\in\C$. 
Fix $P\in\C$, RDF graph $G$ and mapping $\mu$. Let $\F:=\wdpf(P)$ and suppose that $\F=\{\T_1,\dots,\T_m\}$. 
As hinted at the beginning of this section, the idea for checking $\mu\in\sem{\F}{G}$ is to apply the natural evaluation algorithm for wdPFs (see e.g. \cite{lete,PS14}), 
but instead of checking whether $\mu$ can be extended to a child via a homomorphism, 
we check whether it can be extended via the existential $(k+1)$-pebble game. 

Formally, we iterate over the set $\{1,\dots,m\}$ starting from $i=1$, and check the existence of a subtree $\T^\mu_i$ of $\T_i$ such that 
$\mu$ is a homomorphism from $\pat{\T^\mu_i}$ to $G$ (in particular, $\vars{\T^\mu_i}=\dom{\mu}$). 
If there is no such a subtree, we continue with $i+1$. 
By condition (3) of wdPTs, the previous check can be done in polynomial time. 
Note also that $\T^\mu_i$ is unique (if it exists), as $\T_i$ is in NR normal form. 
We now check that for all children $n$ of $\T^\mu_i$, it is \emph{not} the case that 
$$(\pat{\T^\mu_i}\cup\pat{n}, \vars{\T^\mu_i})\rightarrow_{k+1}^\mu G.$$
If this holds, we \emph{accept} the instance; otherwise we continue with $i+1$. 
If for every $i\in\{1,\dots,m\}$ the instance is not accepted, then we \emph{reject} the instance. 

Notice that by Proposition \ref{prop:games-ptime} the above-described algorithm can be implemented in polynomial time. 
Observe also that the algorithm is always sound (independently of the assumption that $\C$ is of bounded domination width). 
Indeed, suppose that $\mu\not\in\sem{\F}{G}$. This means that for each $i\in\{1,\dots,m\}$, 
either $\T^\mu_i$ does not exist or there is such a $\T^\mu_i$
and there is a child $n$ and a homomorphism $\nu$ from $\pat{n}$ to $G$ compatible with $\mu$. 
In other words, 
$$(\pat{\T^\mu_i}\cup\pat{n}, \vars{\T^\mu_i}) \rightarrow^\mu G.$$ 
By property (\ref{eq:relax}), we have that 
$$(\pat{\T^\mu_i}\cup\pat{n}, \vars{\T^\mu_i})\rightarrow_{k+1}^\mu G.$$ 
Hence, the algorithm rejects (as it does not accept in any iteration). 

For completeness, assume that $\mu\in\sem{\F}{G}$, i.e., 
$\mu\in\sem{\T_\ell}{G}$, for some $\ell\in\{1,\dots,m\}$. 
In particular, there is a (unique) subtree $\T$ of $\T_\ell$ such that 
$\mu$ is a homomorphism from $\pat{\T}$ to $G$ and, for every child $n$ of $\T$, 
$(\pat{\T}\cup \pat{n}, \vars{\T})\rightarrow^\mu G$ 
does not hold. 
Towards a contradiction suppose that the algorithm rejects the instance $(\F,G,\mu)$. 
Let $I\subseteq{\supp{\T}}$ be the set of indices $i$ such that, in the $i$-th iteration,  
the algorithm finds a subtree $\T^\mu_i$ of $\T_i$ such that $\mu$ is a homomorphism from $\pat{\T^\mu_i}$ to $G$. 
Observe that $I\neq \emptyset$ as $\ell\in I$.  
Since the algorithm rejects, we have that for every $i\in I$, there is a child $n_i$ of the subtree $\T^\mu_i$ such that 
$$(\pat{\T^\mu_i}\cup\pat{n_i}, \vars{\T^\mu_i})\rightarrow_{k+1}^\mu G.\qquad (\dagger)$$
Let $\Delta$ be the children assignment with $\dom{\Delta}=I$ such that $\Delta(i)=n_i$, for every $i\in I$. 

For readability, we let $X:= \vars{\T}$. We show that $\Delta$ is valid. By contradiction, suppose that there exists $j\in\supp{\T}\setminus I$ such that 
$$(\pat{\T^{\tsupp}(j)}, X)\rightarrow (S_\Delta,X).\qquad (\ddagger)$$
Note first that $\vars{\T^\mu_i}=X$, for every $i\in I$, and hence 
$$(\pat{\T^\mu_i}\cup\pat{n_i}, \vars{\T^\mu_i})=(\pat{\T^\mu_i}\cup\pat{n_i}, X).$$
Recall that $\rho_\Delta(i)$ is the renaming of $\pat{n_i}$ where variables from $\vars{n_i}\setminus X$ become fresh variables. 
We have then that, for every $i\in I$, $(\pat{\T^\mu_i}\cup\rho_\Delta(i), X)\rightarrow (\pat{\T^\mu_i}\cup\pat{n_i}, X)$, 
and by ($\dagger$) and Proposition \ref{prop:games}, item (1), $(\pat{\T^\mu_i}\cup\rho_\Delta(i), X)\rightarrow^\mu_{k+1} G$.

Now we can apply Proposition \ref{prop:games}, item (2) to the generalised t-graphs $\{(\pat{\T^\mu_i}\cup\rho_\Delta(i), X)\mid i\in I\}$ 
and obtain that  
$$(\bigcup_{i\in I}\pat{\T^\mu_i}\cup\rho_\Delta(i), X)\rightarrow^\mu_{k+1} G.$$
Recall that $S_\Delta=\pat{\T}\cup\bigcup_{i\in I} \rho_\Delta(i)$, and since $\ell\in I$ and $\T=\T^\mu_\ell$, 
we have $S_\Delta\subseteq \bigcup_{i\in I}\pat{\T^\mu_i}\cup\rho_\Delta(i)$. It follows that 
$$(S_\Delta, X)\rightarrow (\bigcup_{i\in I}\pat{\T^\mu_i}\cup\rho_\Delta(i), X).$$
By Proposition \ref{prop:games}, item (1), we have
$$(S_\Delta, X)\rightarrow^\mu_{k+1}G \qquad (*)$$
and by ($\ddagger$), it follows that $(\pat{\T^{\tsupp}(j)}, X\rightarrow^\mu_{k+1} G$.  
As $\vars{\T^{\tsupp}(j)}\setminus X=\emptyset$, 
we conclude by property (\ref{eq:empty}) that  $(\pat{\T^{\tsupp}(j)}, X)\rightarrow^\mu G$, i.e., $\mu$ is a homomorphism from $\pat{\T^{\tsupp}(j)}$ to $G$. 
Since $j\not\in I$, this is a contradiction with the definition of $I$. Thus $\Delta\in \VCA{\T}$.

Since $\dw{\F}\leq k$, $\GtG{\T}$ is $k$-dominated. In particular, it is the case that 
$\ctw{S_{\Delta'},X}$ $\leq$ $k$ and $(S_{\Delta'},X)$ $\rightarrow$ $(S_\Delta,X)$, for some $\Delta'\in\VCA{\T}$. 
Proposition \ref{prop:games}, item (1) and ($*$) implies $(S_{\Delta'}, X)\rightarrow^\mu_{k+1}G$.  
By Proposition \ref{prop:games-tw}, we have that $(S_{\Delta'}, X)\rightarrow^\mu G$. 
Since $\T^\tsupp(\ell)=\T$, we have that $(\pat{\T^\tsupp(\ell)},X)$ $\rightarrow$ $(S_{\Delta'}, X)$, and since $\Delta'$ is valid,  
it must be the case that $\ell\in \dom{\Delta'}$. Observe that 
$$(\pat{\T}\cup \pat{\Delta'(\ell)}, X)\rightarrow(S_{\Delta'},X)$$
as $S_{\Delta'}$ contains a copy of $\pat{\T}\cup \pat{\Delta'(\ell)}$, modulo renaming of variables in $\vars{\Delta'(\ell)}\setminus X$. 
By composition, we have 
$$(\pat{\T}\cup \pat{\Delta'(\ell)}, X)\rightarrow^\mu G.$$ 
Since $\Delta'(\ell)$ is a child of $\T$, this contradicts  the fact that $\mu\in\sem{\T_\ell}{G}$.
We conclude that the algorithm accepts the instance $(\F,G,\mu)$. 
\end{proof}

We remark that the classes of bounded domination width \emph{strictly} extend those that are locally tractable \cite{lete} (see also \cite{BPS15}), 
which are the most general tractable restrictions known so far. 
In our context, a class $\C$ is \emph{locally tractable} if there is a constant $k\geq 1$ such that for every 
$P\in\C$ with $\wdpf(P)=\F$, every wdPT $\T=(T,r,\lambda)\in \F$, and every node $n\in V(T)$ with $n\neq r$ and parent $n'$, it is the case that 
$$\ctw{\pat{n},\vars{n}\cap\vars{n'}} \leq k.$$

Observe that local tractability implies bounded domination but the converse does not hold in general. 
Indeed, it suffices to consider the class $\C=\{P_k\mid k\geq 2\}$ from Example \ref{ex:dom1}. 
As shown in this example, $\C$ has bounded domination width but due to node $n_{12}$ in $\T_1$ (see Figure \ref{fig:main}), $\C$ is not locally tractable.

\subsection{The case of UNION-free patterns}
\label{sec:union-free}

In this section, we show that for well-designed patterns using only AND and OPT, 
the notion of domination width boils down to a simpler notion of width called \emph{branch treewidth}. 
Recall that, in this case, well-designed patterns can be represented by pattern trees, instead of pattern forests. 

For a wdPT $\T=(T,r,\lambda)$ and $n\in V(T)$, we define the \emph{branch} $\B_n$ of $n$ to be the set of nodes in $V(T)$ appearing in the unique path in $T$ 
from the root $r$ to the parent of $n$. Note that $\B_r=\emptyset$. For $n\in V(T)$ with $n\neq r$, we define the t-graph $S_n^\tbr:=\pat{n}\cup\bigcup_{n'\in \B_n}\pat{n'}$ 
and the set of variables $X_n^\tbr:=\vars{\bigcup_{n'\in \B_n}\pat{n'}}$. Note that  $X_n^\tbr\subseteq\vars{S_n^\tbr}$. 

\begin{definition}[Branch treewidth]
Let $\T$ $=$ $(T,$ $r,$ $\lambda)$ be a wdPT. We define the \emph{branch treewidth} $\bw{\T}$ of $\T$ to be the minimum positive integer $k$ such that for all $n\in V(T)$ with $n\neq r$, 
it is the case that $\ctw{S_n^\tbr,X_n^\tbr}$ $\leq$ $k$. 

For a UNION-free well-designed graph pattern $P$, we define the  \emph{branch treewidth} of $P$ to be $\bw{P}:=\bw{\T}$, where $\T$ is the wdPT such that $\wdpf(P)=\{\T\}$. 
\end{definition}

As it turns out, branch treewidth and domination width coincide for UNION-free patterns. 

\begin{proposition}
\label{prop:branch}
For every UNION-free well-designed graph pattern $P$, we have that $\dw{P}=\bw{P}$.
\end{proposition}

\begin{proof}
Assume that $\wdpf(P)=\{\T\}$, where $\T=(T,r,\lambda)$ is a wdPT. 
We start by proving that $\dw{\T}\leq \bw{\T}$. Assume that $\bw{P}=k$ and let $\T'$ be a subtree of $\T$. 
We need to prove that $\GtG{\T'}$ is $k$-dominated. We shall prove something stronger: for every $(S_\Delta,\vars{\T'})$ $\in$ $\GtG{\T'}$, 
we have $\ctw{S_\Delta,\vars{\T'}}\leq k$.

Let $(S_\Delta,\vars{\T'})$ $\in$ $\GtG{\T'}$, where $\Delta\in\VCA{\T'}$. 
Observe that $S_\Delta$ coincides with $S':=\pat{\T'}\cup\pat{n}$ modulo renaming of variables in $\vars{n}\setminus\vars{\T'}$, 
where $n$ is a child of $\T'$. Thus 
$$\ctw{S_\Delta,\vars{\T'}}=\ctw{S',\vars{\T'}}.$$
Note that $n\neq r$. 
Let $(C,X_n^\tbr)$ be the core of $(S_n^\tbr,X_n^\tbr)$. 
In particular, $(C,X_n^\tbr)$ is a subgraph of $(S_n^\tbr,X_n^\tbr)$ and $(S_n^\tbr,X_n^\tbr)\rightarrow (C,X_n^\tbr)$. 
As $\B_n\subseteq V(T')$, where $\T'=(T',r,\lambda')$, we have that $(\pat{\T'}\cup C, \vars{\T'})$ is a subgraph of $(S',\vars{\T'})$  and 
$$(S',\vars{\T'})\rightarrow(\pat{\T'}\cup C, \vars{\T'}).$$
Then the core of $(S',\vars{\T'})$ is a subgraph of $(\pat{\T'}\cup C, \vars{\T'})$. 
As treewidth does not increase by taking subgraphs, and using the fact that 
$$\tw{\pat{\T'}\cup C, \vars{\T'}}=\tw{C,X_n^\tbr}$$ 
as their Gaifman graphs coincide, we have that 
$$\ctw{S_\Delta,\vars{\T'}}=\ctw{S',\vars{\T'}}\leq \tw{C,X_n^\tbr}.$$
Since $\bw{\T}=k$ and $n\neq r$, we have that $\tw{C,X_n^\tbr}$ $\leq$ $k$, and hence, $\ctw{S_\Delta,\vars{\T'}}$ $\leq$ $k$ as required. 

We now prove that $\bw{\T}\leq \dw{\T}$. Let $\dw{\T}=k$. 
By contradiction, suppose that there exists $n\in V(T)$ with $n\neq r$ such that $\ctw{S_n^\tbr,X_n^\tbr}$ $>$ $k$. 
Let $\T'$ be the subtree of $\T$ corresponding to $\B_n$. In particular, $n$ is a child of $\T'$ and 
$$(\pat{\T'}\cup\pat{n},\vars{\T'})=(S_n^\tbr,X_n^\tbr).$$ 
For readability, we let $S:=S_n^\tbr$ and $X':= X_n^\tbr=\vars{\T'}$. 
Since $\GtG{\T'}$ is $k$-dominated and $\ctw{S,X'}$ $>$ $k$, there exists a child $n'$ of $\T'$ with $n'\neq n$ such that 
$$(S_{n'},X')\rightarrow(S,X')\quad (\dagger)$$
where $S_{n'}:=\pat{\T'}\cup\pat{n'}$ and $\ctw{S_{n'},X'}\leq k$. 
Let $\T''$ be the subtree of $\T$ obtained from $\T'$ by adding the child $n'$. 
Let $S_{nn'}:=\pat{\T''}\cup\pat{n}$ and $X'':=\vars{\T''}$. 
Below we show that 
$$\ctw{S_{nn'},X''}>k.\quad (*)$$

Towards a contradiction, assume that $\ctw{S_{nn'},X''}\leq k$. 
We shall show that $\ctw{S,X'}\leq k$, which is a contradiction. 
It is a known fact (see \cite[Theorem~12]{DKV}) that $\ctw{S,X'}\leq k$ if and only if ($\ddagger$) there exists $(S^*,X')$ such that 
\begin{itemize}
\item $\tw{S^*,X'}\leq k$, and 
\item $(S,X')$ $\rightarrow$ $(S^*,X')$ and $(S^*,X')$ $\rightarrow$ $(S,X')$. In this case, we write $(S,X')$ $\leftrightarrows$ $(S^*,X')$. 
\end{itemize}
Also, observe that ($\dagger$) implies that $(S_{nn'},X')\rightarrow (S,X')$. 
As $S\subseteq S_{nn'}$, we have $(S,X')$ $\rightarrow$ $(S_{nn'},X')$, and hence
$(S_{nn'},X')\leftrightarrows(S,X')$.  
By transitivity of $\rightarrow$, 
it suffices to show ($\ddagger$) with respect to $S_{nn'}$ instead of $S$. 

We have $\ctw{S_{n'},X'}\leq k$ and $\ctw{S_{nn'},X''}$ $\leq$ $k$, by hypothesis. 
Hence $\tw{C_{n'},X'}\leq k$ and $\tw{C_{nn'},X''}\leq k$, where $(C_{n'},X')$ and $(C_{nn'},X'')$ are the cores of $(S_{n'},X')$ and $(S_{nn'},X'')$, respectively. 
We define the following generalised t-graphs:
\begin{align*}
 D_{n'}& := C_{n'}\setminus \pat{\T'}, \\
 D_{nn'}&:= C_{nn'}\setminus \pat{\T''},\\
 S^*&:= \pat{\T'} \cup D_{n'}\cup D_{nn'}. 
\end{align*}
Note that $\vars{D_{n'}}\cap\vars{D_{nn'}}\subseteq X'$. 
In particular, the Gaifman graph of $(S^*,X')$
 is the disjoint union of those of $(C_{n'},X')$ and $(C_{nn'},X'')$, and then $\tw{S^*,X'}\leq k$. 
It suffices to show that $(S_{nn'},X')$ $\leftrightarrows$ $(S^*,X')$. 

Observe first that  $(S^*,X')$ $\rightarrow$ $(S_{nn'},X')$ as $S^*\subseteq S_{nn'}$. 
For the other direction, observe that $(S_{nn'},X'')$ $\rightarrow$ $(C_{nn'},X'')$ by definition of cores. 
Then $(S_{nn'},X')$ $\rightarrow$ $(C_{nn'},X')$. 
Also by definition of cores, we have that $(S_{n'},X')$ $\rightarrow$ $(C_{n'},X')$ via a homomorphism $h$. 
Hence, by construction of $S^*$, the function $g:\vars{C_{nn'}}\rightarrow \bI\cup\bV$ such that $g(?x)=h(?x)$, for $?x\in\vars{n'}\setminus X'$, and $g(?x)=?x$ otherwise, 
is a homomorphism witnessing $(C_{nn'},X')\rightarrow (S^*,X')$. By transitivity, $(S_{nn'},X')\rightarrow (S^*,X')$ as required. 
Thus claim (*) holds.  

As $\GtG{\hat\T}$ is $k$-dominated for every subtree $\hat\T$ of $\T$, 
we can iterate the previous argument until we find a subtree $\T^*$ of $\T$ such that $n$ is its only child and $\ctw{\pat{\T^*}\cup\pat{n},\vars{\T^*}}>k$. 
It follows that $\GtG{\T^*}$ cannot be $k$-dominated; a contradiction. 
\end{proof}

Proposition \ref{prop:branch} tells us that if $\T$ is a wdPT with $\dw{\T}$ $\leq$ $k$, then for every subtree $\T'$ of $\T$, the set $\GtG{\T'}$ is 
$k$-dominated due to the trivial reason: all elements of $\GtG{\T'}$ are already of ctw $\leq k$.  
Observe that this is not the case for arbitrary patterns. 
Indeed, as Example \ref{ex:dom1} shows, $\dw{\F_k}=1$ but the set $\GtG{\T_1[r_1]}$ is not trivially $1$-dominated as $\ctw{S_{\Delta_2},\{?x,?y\}}=k-1$. 

The results in this paper (see Theorem \ref{theo:main} and Corollary \ref{coro:main-branch} in the next section) 
show that domination width (and then branch treewidth for UNION-free patterns) captures tractability for well-designed patterns. 
Therefore, polynomial-time solvability of arbitrary patterns and UNION-free patterns is based on two different principles: for arbitrary patterns is based on $k$-domination, while 
for UNION-free patterns \emph{branch tractability} suffices. 
As a matter of fact, this striking difference between the general and UNION-free case is also present in other contexts: 
for instance, \emph{containment} of UNION-free patterns can be characterised in very simple terms, 
while the general case requires more involved characterisations 
(see e.g. \cite[Theorem~3.7]{PS14} and \cite[Lemma~1]{KRRV15}).  

Finally, observe that bounded branch treewidth implies local tractability, but the converse is not true in general. 
Hence, we obtain new tractable classes even in the UNION-free case. 
To see this, consider for instance the class $\C=\{P'_k\mid k\geq 2\}$, where $\wdpf(P'_k)=\{\T'_k\}$, where $\T'_k=(T,r,\lambda)$ is a wdPT such that
\begin{itemize}
\item $T=(\{r,n_k\}, \{\{r,n_k\}\})$, i.e., $T$ is the tree containing two nodes. 
\item $\lambda(r)=\{(?y,r,?y)\}$ and $$\lambda(n_k)=\{(?y,r,?o_1)\}\cup K_k(?o_1,\dots,?o_k)$$ where $K_k(?o_1,\dots,?o_k)$ is defined as in Example \ref{ex:cores}. 
\end{itemize}
We have that $\C$ has bounded branch treewidth as $\bw{\T'_k}$ $=$ $1$, for every $k\geq 2$. Indeed, 
the core of $(S_{n_k}^\tbr, X_{n_k}^\tbr)$ is simply $(\{(?y,r,?y)\}, \{?y\})$. On the other hand, $\C$ is not locally tractable as  
$\ctw{\pat{n_k},\{?y\}}=k-1$.

\section{A matching hardness result} 
\label{sec:hardness}

We start by giving some basic definitions from parameterised complexity theory as our hardness result relies on it 
(we refer the reader to \cite{FG06} for more details).   

A \emph{parameterised problem} $(\Pi,\kappa)$ is a classical decision problem $\Pi$ equipped with a \emph{parameterisation} $\kappa$ 
that maps instances of $\Pi$ to natural numbers. The class \emph{FPT} contains all parameterised problems $(\Pi,\kappa)$ that are \emph{fixed-parameter tractable}, 
that is, that can be solved in time $f(\kappa(x))\cdot |x|^{O(1)}$, where $|x|$ denotes the size of the instance and $f:\mathbb{N}\rightarrow\mathbb{N}$ is a computable function. 
An \emph{fpt-reduction} from $(\Pi,\kappa)$ to $(\Pi',\kappa')$ is a function $r$ mapping 
instances of $\Pi$ to instances of $\Pi'$ such that (i) for all instance $x$ of $\Pi$, 
 we have $x\in \Pi$ if and only if $r(x)\in \Pi'$, (ii) $r$ can be computed in time $f(\kappa(x))\cdot|x|^{O(1)}$ for some computable function $f:\mathbb{N}\rightarrow\mathbb{N}$, 
 and (iii) there is a computable function $g:\mathbb{N}\mapsto\mathbb{N}$ such that for all instances
$x$ of $\Pi$, we have $\kappa'(r(x))\leq g(\kappa(x))$. 

The class W[1] can be seen as an analogue of NP in parameterised
complexity theory (for a precise definition, see \cite{FG06}). 
Proving W[1]-hardness (under fpt-reductions) is a strong indication that the problem is
not in FPT as it is believed that FPT $\neq$ W[1]. 
A canonical W[1]-complete problem is $p$-{\sc CLIQUE}, that is, the {\sc CLIQUE} problem parameterised by the size of the clique. 
Recall that the {\sc CLIQUE} problem asks, given an undirected graph $H$ and a positive integer $k$, whether 
$H$ contains a clique of size $k$. 

Given a class $\C$ of well-designed graph patterns, we denote by $p$-$\wdeval(\C)$, the problem $\wdeval(\C)$ parameterised by the size $|P|$ of the input well-designed graph pattern $P$. 
We denote by {co-}$\wdeval(\C)$ the complement of $\wdeval(\C)$, i.e., the problem of checking $\mu\not\in\sem{P}{G}$ for a given well-designed pattern $P$, an RDF graph $G$ and 
a mapping $\mu$. Similarly, we denote by $p$-co-$\wdeval(\C)$ the complement of $p$-$\wdeval(\C)$. 

\subsection{Hardness result and main characterisation theorem}

Our main hardness result is as follows. 

\begin{theorem}[Main hardness]
\label{theo:main-hardness}
Let $\C$ be a recursively enumerable class of well-designed graph patterns of unbounded domination width. 
Then p-co-$\wdeval(\C)$ is $\Wone$-hard. 
\end{theorem}

We provide a proof of Theorem \ref{theo:main-hardness} in the next section. 
We now explain how Theorem \ref{theo:main-tractab} and \ref{theo:main-hardness} imply the main characterisation result 
of this paper. 

\begin{theorem}[Main]
\label{theo:main}
Assume FPT $\neq$ W[1]. Let $\C$ be a recursively enumerable\footnote{As in \cite{gro07}, we can remove the assumption of $\C$ being recursively enumerable by assuming a stronger assumption than FPT $\neq$ W[1] involving non-uniform complexity classes.}class of well-designed graph patterns. Then, the following are equivalent: 
\begin{enumerate}
\item $\wdeval(\C)$ is in PTIME.
\item p-$\wdeval(\C)$ is in FPT. 
\item $\C$ has bounded domination width. 
\end{enumerate}
\end{theorem}

\begin{proof}
(1)$\Rightarrow$(2) is immediate. For (2)$\Rightarrow$(3), if $p$-$\wdeval(\C)$ is in FPT, then $p$-co-$\wdeval(\C)$ also is. 
Then, by our assumption FPT $\neq$ W[1], $p$-co-$\wdeval(\C)$ cannot be $\Wone$-hard. 
Therefore, $\C$ has bounded domination width, otherwise we reach a contradiction by Theorem \ref{theo:main-hardness}. 
The implication (3)$\Rightarrow$(1) follows directly from Theorem \ref{theo:main-tractab}.
\end{proof}
As a corollary of Proposition \ref{prop:branch}, we have the following. 
\begin{corollary}
\label{coro:main-branch}
Assume FPT $\neq$ W[1]. Let $\C$ be a recursively enumerable class of UNION-free well-designed graph patterns. Then, the following are equivalent: 
\begin{enumerate}
\item $\wdeval(\C)$ is in PTIME.
\item p-$\wdeval(\C)$ is in FPT. 
\item $\C$ has bounded branch treewidth. 
\end{enumerate}
\end{corollary}

\subsection{Proof of Theorem \ref{theo:main-hardness}}
\label{sec:proof-hardness}

We follow a similar strategy of the classical result by Grohe \cite{gro07} that shows $\Wone$-hardness for 
evaluating a class $\C$ of CQs over schemas of \emph{bounded arity} whose cores have unbounded treewidth. 
As in \cite{gro07}, we exhibit an fpt-reduction from $p$-CLIQUE to  $p$-co-$\wdeval(\C)$ exploiting 
the Excluded Grid Theorem \cite{Robertson86:excluding} that states that there exists a function ${\tt w}:\mathbb{N}\to \mathbb{N}$ 
such that for every $k\geq 1$, the \emph{$(k\times k)$-grid} is a \emph{minor} of every graph of treewidth at least ${\tt w}(k)$ (see \cite{diestel} for technical details). 
Throughout this section, we use ${\tt w}$ to denote such a function. 

The first ingredient in our proof is the following variant of the main construction from \cite{gro07} to take distinguished elements into account. (See the appendix for a proof.)

\begin{lemma} 
\label{lemma:groheB}
Let $k\geq 2$ and $H$ be an undirected graph. 
Let $(S,X)$ be a generalised t-graph with $\ctw{S,X}\geq {\tt w}({k\choose 2})$. 
Then there is a generalised t-graph $(B,X)$ such that
\begin{enumerate}
\item if $t\in S$ and $\vars{t}\subseteq X$, then $t\in B$. 
\item $(B,X)\rightarrow (S,X)$. 
\item $H$ contains a clique of size $k$ iff $(S,X)\rightarrow(B,X)$. 
\item $(B,X)$ can be computed in time $f(k,|(S,X)|)\cdot |H|^{O(1)}$, where $f:\mathbb{N}\times\mathbb{N}\rightarrow\mathbb{N}$ is a 
computable function. 
\end{enumerate}
\end{lemma}

The second ingredient is the following basic property of wdPFs of large domination width. Intuitively, 
it states that every wdPF $\F$ of large domination width contains a subtree $\T$ with an associated generalised t-graph $(S,X)\in \GtG{\T}$ of large $\ctw{S,X}$,  
satisfying a particular minimality condition. 

\begin{lemma}
\label{lemma:large-dw}
Let $k\geq 2$ and $\F$ be a wdPF such that $\dw{\F}\geq k$. 
Then there exists a subtree $\T$ of $\F$, and $(S,\vars{\T})\in\GtG{\T}$ such that 
\begin{enumerate}
\item $\ctw{S,\vars{\T}}\geq k$, and 
\item whenever $(S',\vars{\T})$ $\rightarrow$ $(S,\vars{\T})$ holds, then $(S,\vars{\T})$ $\rightarrow$ $(S',\vars{\T})$ also holds, for every $(S',\vars{\T})\in\GtG{\T}$. 
\end{enumerate}
\end{lemma}

\begin{proof}
Suppose that $\dw{\F}\geq k$, i.e., $\dw{\F}\leq k-1$ does not hold. 
By definition of domination width, there is a subtree $\T$ of $\F$ such that $\GtG{\T}$ is not 
$(k-1)$-dominated. In particular, the following subset $\G\subseteq \GtG{\T}$ is non-empty: 
$(R,\vars{\T})\in\G$ if and only if  $(R,\vars{\T})\in \GtG{\T}$, $\ctw{R,\vars{\T}}\geq k$, and $(R',\vars{\T})$ $\not\rightarrow$ $(R,\vars{\T})$, 
for all $(R',\vars{\T})\in\GtG{\T}$ with $\ctw{R',\vars{\T}}$ $\leq$ $k-1$. 
Consider the directed graph $H$ with vertex set $\G$ and the existence of homomorphism relation $\rightarrow$ as the edge relation. 
Let $C$ be a minimal strongly connected component of $H$ and pick any $(S,\vars{\T})\in C$. 
We claim that $(S,\vars{\T})$ satisfies the required conditions.  
Indeed, suppose that $(S',\vars{\T})\rightarrow (S,\vars{\T})$. Since $(S,\vars{\T})\in \G$, and by construction of $\G$, 
it must be the case that 
$(S',\vars{\T})$ $\in$ $\G$. Since $C$ is minimal, $(S',\vars{\T})\in C$, and then there is 
a directed path from $(S,\vars{\T})$ to $(S',\vars{\T})$ in $H$. By transitivity of the relation $\rightarrow$, 
we have $(S,\vars{\T})\rightarrow (S',\vars{\T})$ as required. 
\end{proof}

\medskip
\noindent
{\bf The reduction.} We now present an fpt-reduction from $p$-CLIQUE to $p$-co-$\wdeval(\C)$. 
Let $k\geq 2$ and $H$ be an undirected graph. We start by enumerating the class $\C$ until we find some $P\in\C$ such that 
$\dw{P}\geq {\tt w}({k\choose 2})$. Since $\C$ has unbounded domination width, this is always possible. 
Since the domination width is computable, we can find $P$ in time $\alpha(k)$, for a computable function $\alpha:\mathbb{N}\rightarrow \mathbb{N}$. 
Let $\F:=\wdpf(P)$. Since $\dw{\F}\geq {\tt w}({k\choose 2})$, we can apply Lemma \ref{lemma:large-dw} to obtain a subtree $\T$ 
of $\F$ and $(S,\vars{\T})\in\GtG{\T}$ satisfying the conditions of the lemma. By condition (1), $\ctw{S,\vars{\T}}$ $\geq$ ${\tt w}({k\choose 2})$ and hence, 
by Lemma \ref{lemma:groheB}, we can compute in time $f(k,|(S,\vars{\T})|)\cdot |H|^{O(1)}$ a generalised t-graph $(B,\vars{\T})$ satisfying the conditions in the lemma. 
Observe that $(S,\vars{\T})$ only depends on $k$ and thus, $(B,\vars{\T})$ can be computed in time $g(k)\cdot |H|^{O(1)}$, for some computable function $g$. 

Now we define an RDF graph $G$ and a mapping $\mu$ with $\dom{\mu}=\vars{\T}$.  
The idea is that $G$ is precisely $B$ but interpreted as an RDF graph, i.e., we freeze the variables of $B$, which now become IRIs, and 
$\mu$ is the identity mapping over $\vars{\T}$, modulo freezing of variables in $B$ (note that $\vars{\T}\subseteq \vars{B}$). 
Formally, for $?x\in \vars{B}$, we define $a_{?x}$ to be an IRI. We define $\Psi:\vars{B}\rightarrow \bI$ to be the mapping that maps each $?x\in\vars{B}$ to $a_{?x}$. 
Let $G$ be the RDF graph defined by the set $G:=\{\Psi(t)\mid t\in B\}$ 
and let $\mu$ be the mapping with $\dom{\mu}=\vars{\T}$ such that $\mu(?x)=\Psi(?x)$, for every $?x\in \vars{\T}$. 
By construction, $\Psi$ is a homomorphism from $B$ to $G$ and $(B,\vars{\T})\rightarrow^\mu G$. 
We also define a function $\Theta:\dom{G}\to \bI\cup \bV$, where $\dom{G}\subseteq \bI$ is the set of IRIs appearing in $G$, 
such that $\Theta(a)=?x$ if $a=a_{?x}$ and $\Theta(a)=a$ otherwise. 

Observe that $|P|\leq \alpha(k)$ and that $(P,G,\mu)$ can be computed in fpt-time from $(H,k)$, that is, 
in time $g'(k)\cdot |H|^{O(1)}$ for some computable function $g'$. It remains to show that our reduction is correct, that is, 
$H$ contains a clique of size $k$ if and only if $\mu\not\in \sem{P}{G}=\sem{\F}{G}$. 

\medskip
\noindent
{\bf Correctness of the reduction.} Suppose first that $H$ contains a clique of size $k$. Assume $\F=\{\T_1,\dots,\T_m\}$ and $(S,\vars{\T})$ $=$ $(S_\Delta,\vars{\T})$, 
for some $\Delta\in\VCA{\T}$. Let $\T'_\ell$ be a subtree of $\T_\ell$, with $\ell\in\{1,\dots,m\}$, such that 
$\mu$ is a homomorphism from $\pat{\T'_\ell}$ to $G$. We claim that 
there is a child $n$ of $\T'_\ell$ such that 
$$(\pat{\T'_\ell}\cup\pat{n},\vars{\T'_\ell})\rightarrow^\mu G.$$ 
Note that this implies that $\mu\not\in\sem{\T_\ell}{G}$, and since $\ell$ is arbitrary, it follows that $\mu\not\in\sem{\F}{G}$ as required. 
We prove first that $\ell\in\dom{\Delta}$. Note that $\Theta\circ\mu$ is a homomorphism from $\pat{\T'_\ell}$ to $B$. 
By definition of $\mu$, we have that $(\pat{\T'_\ell},\vars{\T})\rightarrow (B,\vars{\T})$. 
By item (2) in Lemma \ref{lemma:groheB}, it follows that $(\pat{\T'_\ell},\vars{\T})$ $\rightarrow$ $(S_\Delta,\vars{\T})$. 
Since $\T'_\ell=\T^\tsupp(\ell)$ and $\Delta$ is valid, it must be the case that $\ell\in \dom{\Delta}$. 

Recall that $S_\Delta=\pat{\T}\cup\bigcup_{i\in\dom{\Delta}}\rho_\Delta(i)$, where $\rho_\Delta(i)$ 
is obtained from $\pat{\Delta(i)}$ by renaming the variables in $\vars{\Delta(i)}\setminus \vars{\T}$ to fresh variables. 
Since $H$ contains a clique of size $k$, we obtain from Lemma \ref{lemma:groheB}, item (3) that 
$(S_\Delta,\vars{\T})\rightarrow (B,\vars{\T})$. Since $(B,\vars{\T})\rightarrow^\mu G$, we have that $(S_\Delta,\vars{\T})\rightarrow^\mu G$. 
In particular, there is a homomorphism $\nu$ from $\rho_\Delta(\ell)$ to $G$ compatible with $\mu$. 
It follows that there is a homomorphism $\nu'$ from $\pat{\Delta(\ell)}$ to $G$ compatible with $\mu$. 
By considering $\mu\cup\nu'$, we have that 
$$(\pat{\T'_\ell}\cup\pat{\Delta(\ell)},\vars{\T})\rightarrow^\mu G.$$
As $\Delta(\ell)$ is a child of $\T^\tsupp(\ell)=\T'_\ell$, the claim follows. Thus $\mu\not\in\sem{\F}{G}$. 

Assume now that $\mu\not\in\sem{\F}{G}$. Let $I\subseteq\supp{\T}$ such that $i\in I$ if and only if 
$\mu$ is a homomorphism from $\pat{\T^\tsupp(i)}$ to $G$. 
We claim that $I\neq \emptyset$. Since $\T$ is a subtree of $\F$, it suffices to show that $\mu$ is a homomorphism from $\pat{\T}$ to $G$. 
To see this, let $t\in \pat{\T}$. In particular, $t\in S_\Delta$ and $\vars{t}\subseteq \vars{\T}$. We can invoke item (1) in Lemma \ref{lemma:groheB} 
and obtain that $t\in B$. By definition of $G$, $\Psi(t)\in G$, and since $\mu(t)=\Psi(t)$, it follows that $\mu(t)\in G$. Then  
$\mu$ is a homomorphism from $\pat{\T}$ to $G$ and $I\neq \emptyset$. 

Since $\mu\not\in\sem{\F}{G}$, 
for every $i\in I$, there exists a child $n_i$ of $\pat{\T^\tsupp(i)}$ and a homomorphism $\nu_i$ from $\pat{n_i}$ to $G$ compatible with $\mu$. 
Let $\Delta'$ be the children assignment with $\dom{\Delta'}=I$ such that $\Delta'(i)=n_i$, for every $i\in I$. 
It follows that, for every $i\in I$, there is a homomorphism $\nu'_i$ from $\rho_{\Delta'}(i)$ to $G$ compatible with $\mu$. 
By definition of $S_{\Delta'}$, the mapping $h=\mu\cup\bigcup_{i\in I}\nu'_i$ is well-defined and is a homomorphism from $S_{\Delta'}$ to $G$. 
In particular,  $(S_{\Delta'},\vars{\T})\rightarrow^\mu G$. 
We now show that $\Delta'$ is valid. By contradiction, assume that there is $j\in\supp{\T}\setminus I$ such that 
$(\pat{\T^\tsupp(j)},\vars{\T})$ $\rightarrow$ $(S_{\Delta'},\vars{\T})$. 
Since $(S_{\Delta'},\vars{\T})\rightarrow^\mu G$, we have that 
$$(\pat{\T^\tsupp(j)},\vars{\T})\rightarrow^\mu G.$$ 
In particular, $\mu$ is a homomorphism from $\pat{\T^\tsupp(j)}$ to $G$, which contradicts the definition of $I$. 
Hence $\Delta'\in \VCA{\T}$ and consequently $(S_{\Delta'},\vars{\T})\in \GtG{\T}$. 

Observe that, by considering $\Theta\circ h$, $(S_{\Delta'},\vars{\T})$ $\rightarrow$ $(B,\vars{\T})$. 
By item (2) of Lemma \ref{lemma:groheB}, $(B,\vars{\T})$ $\rightarrow$ $(S_\Delta,\vars{\T})$, and hence, $(S_{\Delta'},\vars{\T})$ $\rightarrow$ $(S_\Delta,\vars{\T})$. 
Since $(S_{\Delta'},\vars{\T})$ $\in$ $\GtG{\T}$ and by item (2), Lemma \ref{lemma:large-dw}, we have that 
$(S_{\Delta},\vars{\T})$ $\rightarrow$ $(S_{\Delta'},\vars{\T})$, and then $(S_{\Delta},\vars{\T})$ $\rightarrow$ $(B,\vars{\T})$. 
We can apply item (3) of Lemma \ref{lemma:groheB} and conclude that 
$H$ contains a clique of size $k$ as required.

\section{Conclusions}
\label{sec:conclusions}

We have introduced the notion of domination width for well-designed graph patterns. 
We showed that patterns with bounded domination width can be evaluated in polynomial time (Theorem \ref{theo:main-tractab}). 
In a matching hardness result, we showed that classes of unbounded domination width cannot be evaluated in polynomial time (Theorem \ref{theo:main-hardness}), 
unless a widely believed assumption from parameterised complexity fails. 
This provides a complete complexity classification for the evaluation problem restricted to admissible classes of well-designed graph patterns (Theorem \ref{theo:main}). 

A possible direction for future work is to additionally consider the FILTER and SELECT operators (for a formal semantics of these operators, we refer the reader to \cite{PAG09,PS14}).  
We remark, however, that a complete characterisation of the tractable restrictions seems challenging in these cases. 
Indeed, observe that our classification of Theorem \ref{theo:main} is based on the following dichotomy: 
either co-$\wdeval(\C)$ is in PTIME or it is $\Wone$-hard. 
As we explain below, it is known that this dichotomy fails if we add FILTER or SELECT, in the sense that there is a class $\C$ of queries such that co-$\wdeval(\C)$ is in FPT but is NP-hard. 

For the case of FILTER, we note that well-designed patterns using the FILTER operator can express CQs with \emph{inequalities}. 
Consequently, for each class of undirected graphs $\HH$, it is possible to construct a class $\C_\HH$ of well-designed patterns using AND, OPT and FILTER such that 
co-$\wdeval(\C_\HH)$ is polynomial-time equivalent to the \emph{embedding} problem $\emb(\HH)$ for $\HH$. 
In $\emb(\HH)$, we are given two undirected graphs $H$ and $H'$, where $H\in \HH$, 
and the question is whether there is an embedding, i.e., an injective homomorphism from $H$ to $H'$. 
It is known, for instance, that $\emb(\PP)$ (and consequently co-$\wdeval(\C_\PP)$) is in FPT but is NP-hard, 
where $\PP$ is the class of all paths (see e.g. \cite[Section 8]{gro07} and \cite[Section 13.3]{FG06} for more details). 

For SELECT (or \emph{projection}), it was recently shown in \cite{KPS16} that the evaluation problem for the so-called classes of patterns using AND, OPT and SELECT of bounded \emph{global} treewidth 
and \emph{semi-bounded interface} is in FPT (see \cite[Theorem~5]{KPS16}) but NP-hard (as pointed out in \cite{KPS16}, NP-hardness already follows from results in \cite{BPS15}).  

While the above discussion suggests that obtaining a precise characterisation of the tractable classes  in the presence of FILTER or SELECT could be difficult, 
an interesting research direction would be to characterise the classes that are fixed-parameter tractable. 
In a recent unpublished manuscript \cite{MS17}, this problem was studied for (not necessarily well-designed) pattern trees with projection and 
several complexity classifications were obtained. 
Their work differs to ours in that they consider more expressive patterns and aim for fixed-parameter tractability while we consider simpler patterns but deal with polynomial-time tractability. 
Regarding the FILTER operator, let us remark that obtaining characterisations for fixed-parameter tractability in the presence of FILTER would require to solve a known open problem, 
namely, the corresponding characterisation for problems of the form $\emb(\HH)$ (for further details and recent results, see e.g. \cite{gro07,CGL17,L15}).

It would be also interesting to obtain similar structural characterisations for other variants of the evaluation problem such as the problem of \emph{counting} the number of solutions 
or \emph{enumerating} all solutions (see e.g. \cite{KPS16,PSamw14}); or for fragments beyond the well-designed one such as the class of \emph{weakly} well-designed queries \cite{KK16}.

Finally, note that, related to our results, we have the \emph{recognisition} problem: given a well-designed graph pattern $P$, decide whether $\dw{P}\leq k$ (we assume $k\geq 1$ to be fixed). 
Observe that Proposition \ref{prop:branch} gives us an NP upper bound for this problem in the case of UNION-free patterns (as checking bw $\leq k$ is in NP). 
Also, by using the fact that checking whether a relational structure has a core of treewidth at most $k$ is NP-complete \cite[Theorem~13]{DKV}, we obtain that the recognition problem 
for UNION-free patterns is actually NP-complete. 
For arbitrary well-designed graph patterns, it is possible to obtain a $\Pi^p_2$ upper bound from the definition of domination width. 
It remains an open question whether this $\Pi^p_2$ bound is tight.

\bibliographystyle{abbrv}
\bibliography{bibliography}

\section{Appendix}

\subsection{Proof of Lemma \ref{lemma:groheB}}

\noindent
{\textsc{Lemma 2.}}
\emph{
Let $k\geq 2$ and $H$ be an undirected graph. 
Let $(S,X)$ be a generalised t-graph with $\ctw{S,X}\geq {\tt w}({k\choose 2})$. 
Then there is a generalised t-graph $(B,X)$ such that
\begin{enumerate}
\item if $t\in S$ and $\vars{t}\subseteq X$, then $t\in B$. 
\item $(B,X)\rightarrow (S,X)$. 
\item $H$ contains a clique of size $k$ iff $(S,X)\rightarrow(B,X)$. 
\item $(B,X)$ can be computed in time $f(k,|(S,X)|)\cdot |H|^{O(1)}$, where $f:\mathbb{N}\times\mathbb{N}\rightarrow\mathbb{N}$ is a 
computable function. 
\end{enumerate}
}

\medskip

We devote this section to prove this lemma. Our proof is a simple modification of the main construction of \cite{gro07} to handle distinguished elements. 

We start with some definitions.  For $k,\ell\geq 1$, the \emph{$(k\times \ell)$-grid} is the undirected graph with vertex set $\{1,\dots,k\}\times\{1,\dots,\ell\}$ and an 
edge between $(i,j)$ and $(i',j')$ if $|i-i'|+|j-j'|=1$. It is a known fact that the $(k\times k)$-grid has treewidth $k$ (see e.g. \cite{diestel}). 
We say that an undirected graph $H=(V,E)$ is a \emph{minor} of $H'=(V',E')$ if there is a \emph{minor map} from $H$ to $H'$, that is, a function $\gamma$ mapping each vertex of $H$ to a a non-empty set of vertices in $H'$ such that 
(i) $\gamma(u)$ is connected for all $u\in V$, (ii) for all $u,v\in V$ with $u\neq v$, the sets $\gamma(u)$ and $\gamma(v)$ are disjoint, 
and (iii) for all edges $\{u,v\}\in E$, there is an edge $\{u',v'\}\in E'$ such that $u'\in\gamma(u)$ and $v'\in\gamma(v)$. 
We say that the minor map is \emph{onto} if $\bigcup_{u\in V}\gamma(u)=V'$. Observe that if there is a minor map from $H$ to $H'$, and $H'$ is connected, 
then there is a minor map of $H$ onto $H'$.

Let $k\geq 2$, $H=(V,E)$ be an undirected graph, and $(S,X)$ be a generalised t-graph with $\ctw{S,X}\geq {\tt w}({k\choose 2})$. 
From now on, we let $K:={k \choose 2}$. 
Let $(C,X)$ be the core of $(S,X)$ and $G(C,X)$ be the Gaifman graph of $(C,X)$. 
Suppose that $F_1,\dots,F_r$ are the connected components of  $G(C,X)$. As $\tw{G(C,X)}$ $\geq$ ${\tt w}(K)$, there is a
connected component $F_i$ with $\tw{F_i}\geq {\tt w}(K)$. Without loss of generality, we assume that $F_i=F_1$. 
By the Excluded Grid Theorem, since $\tw{F_1}\geq {\tt w}(K)$, it follows that the $(K\times K)$-grid is a minor of $F_1$, 
and hence $(k\times K)$-grid is a minor of $F_1$. Let $\gamma$ be a minor map from the $(k\times K)$-grid onto $F_1$. 

We fix a bijection $\rho$  between $\{1,\dots,K\}$ and all unordered pairs of elements of $\{1,\dots,k\}$. 
For $p\in\{1,\dots,K\}$, we shall abuse notation and write $p$ instead of $\rho(p)$ and $i\in p$ instead of $i\in \rho(p)$, for $i\in \{1,\dots,k\}$. 
We define the following set $\V\subseteq \bV$ of variables: $?(v,e,i,p,?a)\in \V$ iff $v\in V$, $e\in E$, $i\in\{1,\dots,k\}$, $p\in \{1,\dots,K\}$, $?a\in \gamma(i,p)$ and 
$v\in e$ $\iff$ $i\in p$.

 We denote by $V(F_1)\subseteq \bV$ the vertex set of $F_1$. 
Let $\Pi:(\V\cup \vars{C})\to V(F_1)$ be the mapping such that $\Pi(?(v,e,i,p,?a))=?a$, for all $?(v,e,i,p,?a)\in\V$, and $\Pi(?x)=?x$, for $?x\in \vars{C}$. 
We define 
$$Tr:=\{t\in (\bI\cup \bV)^3\mid \vars{t}\setminus X\subseteq \V \text{ and } \Pi(t)\in C\}.$$
We also define $Tr'\subseteq Tr$ as follows. For $t\in Tr$, we have $t\in Tr'$ iff  ($\dagger$) for
all $?x,?x'\in \vars{t}\setminus X \subseteq \V$. 
 \begin{enumerate}
 \item if $?x=?(v,e,i,p,?a)$ and $?x'=?(v',e',i,p',?a')$, then $v=v'$, and 
 \item if $?x=?(v,e,i,p,?a)$ and $?x'=?(v',e',i',p,?a')$, then $e=e'$. 
 \end{enumerate}
 
Let also $Tr_0:=\{t\in C\mid \vars{t}\setminus X\not\subseteq V(F_1)\}$. 
Then we define $B=Tr'\cup Tr_0$. 

Note that $\Pi$ is a homomorphism from $B$ to $C$. Indeed, if $t\in Tr'$, then $t\in Tr$, and hence $\Pi(t)\in C$. 
If $t\in Tr_0$, then $\Pi(t)=t\in C$. Since $\Pi(?x)=?x$, for all $?x\in X$, we have $(B,X)\rightarrow (C,X)$. Since $(C,X)\rightarrow (S,X)$, 
we have $(B,X)\rightarrow (S,X)$ and condition (2) in the lemma holds.
For condition (1), let $t\in S$ such that $\vars{t}\subseteq X$. Since $(S,X)\rightarrow (C,X)$, then $t\in C$. In particular,  $\vars{t}\setminus X=\emptyset \subseteq \V$ and, 
 by definition of $\Pi$, we have $\Pi(t)=t\in C$. Hence $t\in Tr$. Since  $\vars{t}\setminus X=\emptyset$, $t\in Tr'$ holds trivially. 
Then $t\in B$ as required. Condition (4) is immediate from the definition of $B$. It remains to verify condition (3). 
We can follow the same arguments as in \cite{gro07}.

Suppose that $H$ contains a clique of size $k$ and let $\{v_1,\dots,v_k\}$ be such a clique. For $p\in K$ with $\rho(p)=\{i,j\}$, where $i,j\in\{1,\dots,k\}$ and $i\neq j$, we let 
$e_p$ be the  edge from $v_i$ to $v_j$. In this case we can define $h:\vars{C}\to \vars{B}$ such that $h(?x)=?x$, if $?x\not\in V(F_1)$, 
and $h(?a)=?(v_i,e_p,i,p,?a)$ if $?a\in V(F_1)$, for $i\in \{1,\dots,k\}$ and $p\in \{1,\dots,K\}$ with $\gamma(i,p)=?a$. 
First note that $v_i\in e_p$ $\iff$ $i\in p$. Note that since $v_i$ and $e_p$ are determined by $i$ and $p$, then for every $t\in C$, it is the case that $h(t)\in B$. 
As  $h(?x)=?x$, if $?x\in X$, we conclude that $(C,X)\rightarrow (B,X)$. Since $(S,X)\rightarrow (C,X)$, then $(S,X)\rightarrow (B,X)$ as required.

Conversely, suppose that $(S,X)\rightarrow (B,X)$. In particular, $(C,X)\rightarrow (B,X)$ via a homomorphism $h$.
We claim that there is a homomorphism $g$ from $C$ to $B$ with $g(?x)=?x$, for all $?x\in X$ such that $\Pi\circ g$ is the identity mapping over $\vars{C}$. Indeed, 
let $s=\Pi\circ h$. Then $s$ is a homomorphism witnessing $(C,X)\rightarrow (C,X)$. Since $(C,X)$ is a core, then $s$ must be a bijection and hence an isomorphism. 
It suffices to conisder $g=h\circ s^{-1}$. 
It follows that for all $i\in \{1,\dots,k\}$, $p\in \{1,\dots,K\}$, $?a$ such that $\gamma(i,p)=?a$, $g(?a)$ is of the form 
$$g(?a)=?(v_{a?},e_{?a},i,p,?a)$$
where $v_{?a}\in e_{?a}$ $\iff$ $i\in p$. By the consistency conditions ($\dagger$) and the connectivity of $F_1$, it follows that (i) $v_{?a}=v_{?a'}$ and $e_{?a}=e_{?a'}$, 
whenever $?a,?a'\in \gamma(i,p)$, (ii) $v_{?a}=v_{?a'}$, if $?a\in \gamma(i,p)$ and $?a'\in \gamma(i,p')$, and (iii) $e_{?a}=e_{?a'}$, if $?a\in \gamma(i,p)$ and $?a'\in \gamma(i',p)$. 
It follows that there are vertices $v_1,\dots,v_k\in V$ and edges $e_1,\dots,e_K$ such that whenever $?a\in \gamma(i,p)$ then $g(?a)=?(v_{i},e_{p},i,p,?a)$. 
By the conditions $v_i\in e_i$ $\iff$ $i\in p$, we have that $\{v_1,\dots,v_k\}$ is a clique in $H$ as required.

\end{document}